\documentclass[%
superscriptaddress,
preprint,
amsmath,amssymb,
aps,
pra,
]{revtex4-2}

\usepackage{graphicx}
\usepackage{dcolumn}
\usepackage{bm}
\usepackage[mathlines]{lineno}
\usepackage{amsthm}
\usepackage[shortlabels]{enumitem}
\usepackage{mathrsfs}

\PassOptionsToPackage{breaklinks}{hyperref}
\usepackage{xurl}
\usepackage{hyperref}

\allowdisplaybreaks

\newtheorem{theorem}{Theorem}[section]
\newtheorem{proposition}[theorem]{Proposition}
\newtheorem{corollary}[theorem]{Corollary}

\newtheorem{conjecture}[theorem]{Conjecture}

\theoremstyle{definition}
\newtheorem{definition}[theorem]{Definition}

\theoremstyle{remark}
\newtheorem*{remark}{Remark}

\begin{document}
\linespread{1.25}\selectfont

\preprint{Project Report}

\title{State-independent all-versus-nothing arguments}

\author{Boseong Kim}
\affiliation{
  Department of Physics, University College London,\\Gower Street, London, WC1E 6BT, United Kingdom
}

\author{Samson Abramsky}
  \email{s.abramsky@ucl.ac.uk}
\affiliation{
  Department of Computer Science, University College London,\\Gower Street, London, WC1E 6BT, United Kingdom
}

\date{\today}

\begin{abstract}
Contextuality is a key feature of quantum information that challenges classical intuitions, providing the basis for constructing explicit proofs of quantum advantage. While a number of evidences of quantum advantage are based on the contextuality argument, the definition of contextuality is different in each research, causing incoherence in the establishment of instant connection between their results. In this report, we review the mathematical structure of sheaf-theoretic contextuality and extend this framework to explain Kochen-Specker type contextuality. We first cover the definitions in contextuality with detailed examples. Then, we state the all-versus-nothing (AvN) argument and define a state-independent AvN class. It is shown that Kochen-Specker type contextuality, or contextuality in a partial closure, can be translated into this framework by the partial closure of observables under the multiplication of commuting measurements. Finally, we compare each case of contextuality in an operator-side view, where the strict hierarchy of contextuality class in a state-side view seems to merge into the state-independent AvN class together with the partial closure formalism. Overall, this report provides a unified interpretation of contextuality by integrating Kochen-Specker type notions into the state-independent AvN argument. The results present novel insights into contextuality, which pave the way for a coherent approach to constructing proofs of quantum advantage.
\end{abstract}

\maketitle

\newpage

\tableofcontents

\newpage

\section{Introduction}

So, why do we need a quantum computer? This question must be one of the most frequently asked questions for those who research quantum computers. The claim of exponential quantum advantage in the Deutsch-Jozsa algorithm \cite{DJ92} and the monumental factoring algorithm by Shor \cite{Sho94} opened an intriguing discussion on quantum advantage. They inspired numerous computational tasks for quantum information systems, for example, the GHZ \cite{GHZ90} and magic square games \cite{CHTW10}, quantum search algorithm \cite{Gro96}, and variational quantum eigensolver (VQE) \cite{PMS+14} for quantum simulation and optimization.

Those cases of quantum advantage are crucial for the development of quantum technologies since they can spark a totally novel area of application, opening new markets and research fields for quantum scientists. However, discovering and proving a new instance of quantum advantage are challenging works due to the complexity and counter-intuitive nature of quantum mechanics.

Contextuality is a useful tool to characterize the specific point where quantum information systems depart from classical systems. Originating from Bell's pioneering work on the non-local behavior of quantum mechanics \cite{Bel64}, Kochen and Specker \cite{KS75} characterized the idea of contextuality based on the problem of a hidden variable theory. Subsequent investigations on quantum observables were made by Mermin \cite{Mer90, Mer93} who first used the term \textit{contextuality} in his articles. Recent attempts to apply mathematical structures to contextuality \cite{AB11, AMB12, Abr14, ABK+15, ABCP17, ABM17, ABC+18, Aas22, Rau13, BDB+17, ORBR17, CSW14, KWB18, KL19} have gained a certain degree of success, producing proofs of quantum advantage in some computational processes \cite{Aas22, BDB+17, KWB18}.

\subsection{Applications of contextuality to quantum advantage}

Contextuality as a resource of quantum advantage was first discussed by Raussendorf \cite{Rau13} in a measurement-based quantum computing model (MBQC). His major results state that an MBQC requires contextuality when it computes a non-linear Boolean function with a certain probability. This idea was developed by Bermejo-Vega et al. \cite{BDB+17}, establishing the explicit proof that contextuality is a necessary resource for a large class of MBQC schemes.

It has also been observed that contextuality plays a crucial role in a gate-based quantum computing model. Bravyi, Gosset, and K\"{o}nig \cite{BGK18} reported the quantum advantage of shallow circuits, which was the first mathematical proof of an unconditional quantum advantage for a certain class of quantum circuits. Here they considered strategies for non-local games recast into the circuit, proving a separation between classical and quantum computational models. Subsequently, this strategy was shown to be noise robust \cite{BGKT20} and extended to distributed non-local games \cite{Aas22}.

Karanjai, Wallman, and Bartlett \cite{KWB18} also developed their own framework of contextuality to prove that the spatial complexity of classical simulation of a quantum measurement process is bounded by contextuality. This result is further developed by Kirby and Love \cite{KL19}, who suggested applying contextuality to evaluate the quantum advantage of the VQE algorithm. They proposed a test to determine whether or not the given objective function for the VQE is contextual, employing the compatibility graph of Pauli operators and the measure, contextual $p$-distance. This test detects the non-classicality of the objective function, thus filtering out classically simulatable procedures. This research has important implications for the application of contextuality to evaluate practical quantum algorithms.

\subsection{Classification of contextuality}

The sheaf-theoretic definition of contextuality has been instrumental in our understanding of contextuality, as it provides precise mathematical structure to the intuitive concept of contextuality. The sheaf-theoretic framework was first proposed by Abramsky and Brandenburger \cite{AB11, Abr14}, where they defined events and distributions on the measurement scenario and identified the sheaf structure of those concepts. Here, one can connect the global distribution to the hidden-variable model, which is well-known for its failure to explain the distinct features of quantum theory. Further discussion by Abramsky, Barbosa, and Mansfield \cite{ABM17} investigated a measure of contextuality. This work opened the way to quantify contextuality in the given quantum scenario.

The subsequent development in the cohomological approach to contextuality also provides a substantial methodology for witnessing contextuality in a given measurement scenario. Abramsky, Mansfield, and Barbosa \cite{AMB12} proposed the approach based on the \v{C}ech cohomology invariant, which leverages powerful tools of sheaf cohomology to detect contextuality in the empirical model. The proposal by Okay, Roberts, Bartlett, and Raussendorf \cite{ORBR17} established the topological approach to identifying contextuality, which has the potential to provide a more refined analysis, although an additional topological structure must be concerned. Those approaches were connected by Aasnæss \cite{Aas22}, complementing generality and completeness in each approach by translating arguments from one to another.

On the other hand, a stronger form of contextuality, namely, the all-versus-nothing (AvN) argument was also characterized by the same group. Abramsky et al. \cite{ABK+15, ABCP17} formalized the logical inconsistency in quantum information systems into an AvN argument referring back to the observation by Mermin \cite{Mer90, Mer93}. This class of contextuality is also observed as an obstruction in the cohomology group in the work by Aasnæss \cite{Aas22}.

While the sheaf-theoretic framework provides the base of arguments in the quantum advantage of MBQCs and shallow circuits, the last case of the application, Refs.~\onlinecite{KWB18, KL19}, roots back to Kochen and Specker's framework on formalizing contextuality, so-called contextuality in a closed sub-theory. This notion seems to describe the same idea with the sheaf-theoretic contextuality, but it is solely based on arguments on operator algebra of observables. Moreover, the interesting point is that they bring the 2-qubit measurement scenario of certain observables into the stronger class of contextuality, while the 2-qubit scenario was previously claimed to be unable to show strong non-locality \cite{ABC+18}. This confliction clearly motivated our research on the classification of contextuality, giving a specific case for the two different points of view. Jumping into the result, it turns out that Kochen-Specker type contextuality actually deals with an abelian partial group generated by the given set of observables. Furthermore, we do find strong contextuality in this generated abelian partial group, \textit{i.e.}, in a partial closure.

In this report, we review the mathematical structure of the sheaf-theoretic formalism of contextuality and refine the framework to characterize the connecting point to Kochen-Specker type contextuality. We first go through strict definitions of measurement scenarios, measurement covers, events, distributions, and empirical models. Under these definitions, we review the connection between contextuality and a sheaf theory. Then, we further develop those ideas to classify strong contextuality and all-versus-nothing (AvN) arguments, which define stronger classes of contextuality. The main proposal of this report is the definition of a state-independent AvN class, which is proven in the text that induces AvN arguments for any state realizing the measurement scenario. Then, we show that Kochen-Specker type contextuality is translated into the state-independent AvN in a partial closure, by considering an abelian partial group structure of a Pauli $n$-group. Finally, we summarize the classes of contextuality in a state-side view and an operator-side view, where we can notice that state-dependent contextual scenarios seem to merge into the state-independent AvN class in a partial closure.

\section{Sheaf-theoretic contextuality}

The \textit{sheaf theoretic} framework of Abramsky and Brandenberger \cite{AB11} provides a mathematical structure to formalize the concept of contextuality. In this section, we go through the definitions used in the sheaf theoretic approach with the example of the Bell scenario. We also discuss the contextual fraction as a useful concept for quantifying contextuality, referring to the paper by Abramsky, Barbosa, and Mansfield \cite{ABM17}.

\subsection{Measurement scenarios}

\begin{definition}
A \textit{measurement scenario} is a triple $\left< X, \mathcal{M}, O \right>$ where:
\begin{itemize}
    \item $X$ is a finite set of \textit{measurements};
    \item $\mathcal{M} \subset \mathcal{P}(X)$ is a family of \textit{measurement contexts}, where each context $C \in \mathcal{M}$ represents a set of measurements that can be performed together;
    \item $O$ is a finite set of \textit{outcomes}.
\end{itemize}
\end{definition}

For example, the well-known Bell scenario, where two experimenters, Alice and Bob, can each choose between performing one of two different measurements, say, $a_1$ or $a_2$ for Alice and $b_1$ or $b_2$ for Bob, obtaining one of two possible outcomes, is represented as follows:
\begin{eqnarray*}
& X = \{ a_1, a_2, b_1, b_2 \}, \quad O = \{ 0, 1 \}, \\
& \mathcal{M} = \left\{ \{ a_1, b_1 \}, \{ a_1, b_2 \}, \{ a_2, b_1 \}, \{ a_2, b_2 \} \right\}.
\end{eqnarray*}
Here, the measurement contexts $C \in \mathcal{M}$ explains that we are allowed to measure either $a_1$ or $a_2$, but not both (same for $b_1$ and $b_2$). We will generally assume a measurement scenario $\left< X, \mathcal{M}, O \right>$ for the rest of this paper. We also refer to Bell-type measurement scenarios in a tuple $(n, k, l)$, where $n$ is the number of agents, $k$ is the number of possible measurements for each agent, and $l$ is the number of possible outcomes. Thus, the example is a $(2,2,2)$ scenario.

From the example, we also capture the idea that there needs some additional information about $\mathcal{M}$ to describe situations in quantum mechanics. While we described $\mathcal{M}$ as consisting of ``sets of measurements that can be performed together,'' it remains unclear how many (or how few) sets are needed fully characterize a given scenario. Hence, we aim to collect a sufficient number of sets to cover all possible measurements, while keeping them as compact as possible to represent each set of compatible measurements. Here comes the definition of the measurement cover.

\begin{definition}
A \textit{measurement cover} $\mathcal{M}$ on the set $X$ of measurements is a family of measurement contexts such that:
\begin{itemize}
    \item $\mathcal{M}$ covers $X$: $\displaystyle \bigcup_{C \in \mathcal{M}} C = X$;
    \item $\mathcal{M}$ is an \textit{anti-chain}: If $C, C' \in \mathcal{M}$ and $C \subset C'$ then $C = C'$.
\end{itemize}
\end{definition}

We think of the anti-chain condition because we shall focus on the \textit{maximal} compatible sets of measurements.

\subsection{Events}

\begin{definition}
Given a set of measurements $X$ and a set of outcomes $O$, an \textit{event} or an \textit{assignment over $U \subset X$} is a function $s: U \rightarrow O$.
\end{definition}

For example, consider the event of observing $a_1$ and $b_1$ in the Bell scenario and obtaining outcomes 0 and 1, each. We represent this event by a tuple $(a_1 \mapsto 0, b_1 \mapsto 1)$, or simply, $(0, 1)$. Here we further focus on the mathematical structure of events. We must strictly distinguish events $(a_1 \mapsto 0, b_1 \mapsto 1)$ and $(a_1 \mapsto 0)$. In the latter case, there is no information about what we measured on Bob's side. It means not only that we do not know the outcome Bob obtained, but also that we even cannot say which measurement either $b_1$ or $b_2$ Bob actually measured. It implies that the events are given with respect to the set $U$ of measurements as well as assigning a value $o \in O$ for each measurement $x \in U$. In fact, we can define this set of events as a categorical functor including the restriction map, which formalizes ``forgetting'' information out of what is concerned in the smaller context.

\begin{definition}
Given a set of measurements $X$ and a set of outcomes $O$, a functor $\mathcal{E}: \mathcal{P}(X)^{\mathsf{op}} \rightarrow \mathbf{Set}$, called the \textit{event sheaf}, is defined as follows:
\begin{itemize}
    \item $\displaystyle \forall U \subset X$, $\mathcal{E}(U) := \prod_{x \in U} O$;
    \item $\forall U, U' \subset X$ and $U \subset U'$, a map $\mathsf{res}^{U'}_{U}: \mathcal{E}(U') \rightarrow \mathcal{E}(U)$ of events is defined by the usual functional restriction $\mathsf{res}^{U'}_{U}(s) = \left. s \right| _{U}$.
\end{itemize}
\end{definition}

\begin{figure}
\includegraphics[width=\textwidth]{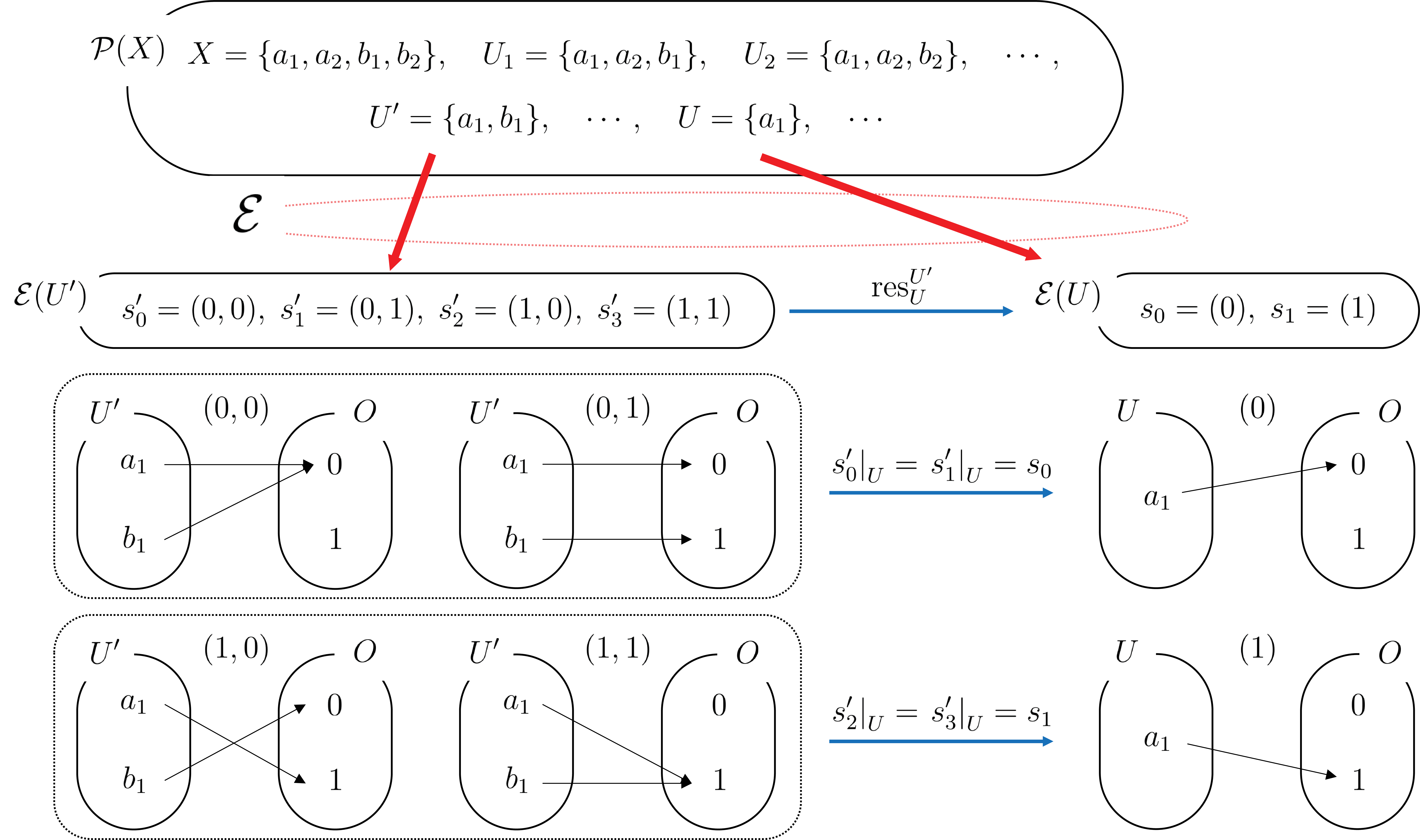}
\caption{\label{fig:Event_sheaf}Visualization of an event sheaf. An event sheaf is a functor consisting of a mapping from $\mathcal{P}(X)$ to a set of events and a restriction mapping between sets of events.}
\end{figure}

\begin{remark}
$\mathbf{Set}$ denotes the category of sets and functions. Here, it is composed of sets of events and restriction maps. $\mathcal{P}(X)$ denotes the power set category of $X$, conventionally composed of sets of subsets of $X$ and inclusion maps. $\mathcal{P}(X)^{\mathsf{op}}$ denotes its opposite, having projection maps instead of inclusion maps. Thus, for $U \subset U'$, a projection $\pi: U' \rightarrow U$ maps to $\mathsf{res}^{U'}_{U}: \mathcal{E}(U') \rightarrow \mathcal{E}(U)$ by $\mathcal{E}$. Refer to Ref.~\onlinecite{HV19} to study categorical notions used for quantum theory.
\end{remark}

Fig.~\ref{fig:Event_sheaf} illustrates an example of an event sheaf defined in the Bell scenario. $\mathcal{E}$ basically maps each subset $U \subset X$ to a set of functions, what we call events on $U$. Each event $s$ is a mapping from $U$ to $O$ such that assigns an outcome for each measurement $x \in U$.

\subsection{Distributions and empirical models}

\begin{definition}
For any set $X$ and a semiring $R$, the \textit{support} $\mathsf{supp}(\phi)$ of a function $\phi: X \rightarrow R$ is defined by:
\begin{equation*}
\mathsf{supp}(\phi) = \left\{ x \in X \mid \phi(x) \neq 0 \right\}.
\end{equation*}
\end{definition}

\begin{definition}
For any set $X$ and a semiring $R$, an \textit{$R$-distribution} on $X$ is a function $d: X \rightarrow R$ which has finite support, and is such that:
\begin{equation*}
\sum_{x \in X} d(x) = 1.
\end{equation*}
\end{definition}

\begin{definition}
For any sets $X, Y$ and a semiring $R$, $\mathcal{D}_{R}$ is a functor defined as follows:
\begin{itemize}
    \item $\mathcal{D}_{R}(X)$ is the set of $R$-distributions on $X$;
    \item given a function $f: X \rightarrow Y$, the action of $\mathcal{D}_{R}$ on $f$ is defined by:
    \begin{equation*}
    \mathcal{D}_{R}(f): \mathcal{D}_{R}(X) \rightarrow \mathcal{D}_{R}(Y):: d \mapsto \left[ y \mapsto \sum_{f(x)=y} d(x) \right].
    \end{equation*}
\end{itemize}
\end{definition}

Now we can compose a functor $\mathcal{D}_R: \mathbf{Set} \rightarrow \mathbf{Set}$ with the event sheaf $\mathcal{E}: \mathcal{P}(X)^{\mathsf{op}} \rightarrow \mathbf{Set}$ to form a functor $\mathcal{D}_R\mathcal{E}: \mathcal{P}(X)^{\mathsf{op}} \rightarrow \mathbf{Set}$. For each $U \subset X$ A distribution $d \in \mathcal{D}_R\mathcal{E}(U)$ maps $\mathcal{E}(U) \rightarrow R$.

\begin{table*}[tbp]
\caption{\label{tab:Empirical_CHSH}Empirical model on $\mathrm{CHSH} - (2, 2, 2)$ scenario. The table is obtained from local projective measurements equatorial at angles 0 for $a_1$, $b_1$ and $\pi/3$ for $a_2$, $b_2$ on the two-qubit Bell state $\left| \phi^+ \right>$.}
\begin{tabular}{c|c c c c}
& $(0,0)$ & $(0,1)$ & $(1,0)$ & $(1,1)$ \\
\colrule
$\{a_1, b_1\}$ & 1/2 & 0   & 0   & 1/2 \\
$\{a_1, b_2\}$ & 3/8 & 1/8 & 1/8 & 3/8 \\
$\{a_2, b_1\}$ & 3/8 & 1/8 & 1/8 & 3/8 \\
$\{a_2, b_2\}$ & 1/8 & 3/8 & 3/8 & 1/8 \\
\end{tabular}
\end{table*}
For example, consider the Bell scenario again. Table~\ref{tab:Empirical_CHSH} realizes the Clauser-Horne-Shimony-Holt (CHSH) model, obtained from local projective measurements equatorial at angles 0 for $a_1$, $b_1$ and $\pi/3$ for $a_2$, $b_2$ on the two-qubit Bell state. In the table, each row represents a context $C$, and each column represents an event $s \in \mathcal{E}(C): C \rightarrow O$ indexed by the outcome for each measurement. Then the first row of the table is a distribution $d \in \mathcal{D}_{R}\mathcal{E}(\{a_1, b_1\})$ mapping each event $s \in \mathcal{E}(\{a_1, b_1\})$ to the semiring $\mathbb{R}_{\geq 0}$ as follows:
\newpage
\begin{eqnarray*}
d((0, 0)) & = & 1/2 \\
d((0, 1)) & = & 0 \\
d((1, 0)) & = & 0 \\
d((1, 1)) & = & 1/2.
\end{eqnarray*}

\begin{remark}
Given $U \subset U'$, $\mathcal{D}_{R}\mathcal{E}$ maps the projection map $\pi: U \rightarrow U'$ to:
\begin{equation*}
\mathcal{D}_{R}\mathcal{E}(U') \rightarrow \mathcal{D}_{R}\mathcal{E}(U):: d \mapsto \left. d \right| _{U},
\end{equation*}
where for each events $s \in \mathcal{E}(U)$:
\begin{equation*}
\left. d \right| _{U} (s) = \sum_{s' \in \mathcal{E}(U'), \left. s' \right| _{U} = s} d(s').
\end{equation*}
Thus $\left. d \right| _{U}$ is the \textit{marginal} of the distribution $d$, which assigns to each event $s$ in the smaller context $U$ the sum of weights of all events $s'$ in the larger context which restrict to $s$.
\end{remark}

In the later part of this paper, it turns out to be obvious that the functor $\mathcal{D}_R\mathcal{E}$ cannot be a sheaf on the semirings we are concerned. However, we actually deal with a specific family of distributions when we look into a quantum system. This concept is characterized by an \textit{empirical model}.

\begin{definition}
Given a measurement cover $\mathcal{M}$, a \textit{no-signaling empirical model} $e$ for $\mathcal{M}$ is a compatible family of distributions such that:
\begin{itemize}
    \item for any measurement context $C \in \mathcal{M}$, $e_C$ is a local distribution of $\mathcal{D}_R\mathcal{E}$ at $C$, \textit{i.e.} $e_C \in \mathcal{D}_R\mathcal{E}(C)$;
    \item $e$ is \textit{compatible} for any measurement contexts $C, C' \in \mathcal{M}$, \textit{i.e.} $\left. e_C \right| _{C \cap C'} = \left. e_{C'} \right| _{C \cap C'}$.
\end{itemize}
The compatibility condition corresponds to the concept no-signaling in the sense that the choice of context, $C$ or $C'$, does not affect the local distribution at $C \cap C'$.
\end{definition}

Now for our Bell scenario, we can show that it is an empirical model $e$ for a measurement cover $\mathcal{M} = \left\{ \{ a_1, b_1 \}, \{ a_1, b_2 \}, \{ a_2, b_1 \}, \{ a_2, b_2 \} \right\}$ where $e_C$ is the distribution specified by each row of the table. The compatibility condition is verified by computing the marginal, for example, for the first two rows:

\begin{eqnarray*}
\left. e_{\{ a_1, b_1 \}} \right| _{a_1}((0)) & = &
\sum_{\left. s' \right| _{a_1} = ((0))} e_{\{ a_1, b_1 \}}(s') \\
& = & e_{\{ a_1, b_1 \}}((0, 0)) + e_{\{ a_1, b_1 \}}((0, 1)) \\
& = & 1/2 + 0 = 1/2; \\
= \left. e_{\{ a_1, b_2 \}} \right| _{a_1}((0)) & = &
\sum_{\left. s' \right| _{a_1} = ((0))} e_{\{ a_1, b_2 \}}(s') \\
& = & e_{\{ a_1, b_2 \}}((0, 0)) + e_{\{ a_1, b_2 \}}((0, 1)) \\
& = & 3/8 + 1/8 = 1/2; \\
\left. e_{\{ a_1, b_1 \}} \right| _{a_1}((1)) & = &
\sum_{\left. s' \right| _{a_1} = ((1))} e_{\{ a_1, b_1 \}}(s') \\
& = & e_{\{ a_1, b_1 \}}((1, 0)) + e_{\{ a_1, b_1 \}}((1, 1)) \\
& = & 0 + 1/2 = 1/2; \\
= \left. e_{\{ a_1, b_2 \}} \right| _{a_1}((1)) & = &
\sum_{\left. s' \right| _{a_1} = ((1))} e_{\{ a_1, b_2 \}}(s') \\
& = & e_{\{ a_1, b_2 \}}((1, 0)) + e_{\{ a_1, b_2 \}}((1, 1)) \\
& = & 1/8 + 3/8 = 1/2.
\end{eqnarray*}
Thus, we verify that $\left. e_{\{ a_1, b_1 \}} \right| _{a_1} = \left. e_{\{ a_1, b_2 \}} \right| _{a_1}$. Here we denoted the restriction map to the one-point set $\{m\}$ by $\left. s \right| _{m}$ rather than $\left. s \right| _{\{m\}}$. We shall keep this abbreviation throughout this paper since it does not violate our intuition too much. The same computation for the other contexts $C \in \mathcal{M}$ shows that $e$ is a no-signaling empirical model. We shall also use the term \textit{empirical model} for the no-signaling model as we will only focus on quantum-like models with the no-signaling condition.

\subsection{Sheaf condition and global distributions}

We previously defined two functors $\mathcal{E}$ and $\mathcal{D}_{R}\mathcal{E}$. In fact, these two functors can be identified as presheaves. $\mathcal{E}$ can be further identified to be a sheaf, justifying its name, 'event sheaf.' Here, we will look through the definitions of presheaves and sheaves referring to Ref \cite{Har77} to make those facts clear. Remark that we can consider the measurement cover $\mathcal{M}$ as a generating set of the topology of $X$, giving each context $C$ an open subset of $X$.

\begin{definition}
Let $X$ be a topological space. A \textit{presheaf} $\mathscr{F}$ on $X$ consists of the data:
\begin{enumerate}[(i)]
    \item for every open subset $U \subset X$, an object $\mathscr{F}(U)$;
    \item for every inclusion $U \subset U'$ of open subsets of $X$, a morphism of objects $\rho_{U}^{U'}: \mathscr{F}(U') \rightarrow \mathscr{F}(U)$;
\end{enumerate}
subject to the conditions:
\begin{enumerate}[(a)]
    \item $\rho_{U}^{U}$ is the identity map $\mathscr{F}(U) \rightarrow \mathscr{F}(U)$;
    \item if $U \subset U' \subset U''$ are three open subsets, then $\rho_{U}^{U''} = \rho_{U}^{U'} \circ \rho_{U'}^{U''}$.
\end{enumerate}
\end{definition}

\begin{definition}
A presheaf $\mathscr{F}$ on a topological space $X$ is a \textit{sheaf} if it satisfies the following supplementary conditions:
\begin{enumerate}[(a)]
    \setcounter{enumi}{2}
    \item (\textit{Locality}) if $U$ is an open set, if $\left\{ V_{i} \right\}$ is an open covering of $U$, and if $s, t \in \mathscr{F}(U)$ are elements such that $\left. s \right| _{V_{i}} = \left. t \right| _{V_{i}}$ for all $i$, then $s = t$;
    \item (\textit{Gluing}) if $U$ is an open set, if $\left\{ V_{i} \right\}$ is an open covering of $U$, and if we have elements $s_{i} \in \mathscr{F}(V_{i})$ for each $i$, with the property that for each $i,j$, $\left. s_{i} \right| _{V_{i} \cap V_{j}} = \left. s_{j} \right| _{V_{i} \cap V_{j}}$, then there is an element $s \in \mathscr{F}(U)$ such that $\left. s \right| _{V_{i}} = s_{i}$ for each $i$.
\end{enumerate}
\end{definition}

In general, we assume discrete topology for the set of measurements $X$. It is easy to show that $\mathcal{E}$ is a sheaf and $\mathcal{D}_{R}\mathcal{E}$ is a presheaf. However, $\mathcal{D}_{R}\mathcal{E}$ generally fails to achieve the sheaf conditions.

\begin{table*}[htbp]
\caption{\label{tab:Locality}Two different distributions over the same measurement context, which is compatible in each local measurement.}
\begin{tabular}{c|c c c c}
& $(0,0)$ & $(0,1)$ & $(1,0)$ & $(1,1)$ \\
\colrule
$d_1$ & 1/2 & 0   & 0   & 1/2 \\
$d_2$ & 1/4 & 1/4 & 1/4 & 1/4 \\
\end{tabular}
\end{table*}
For the counter-example of locality, consider two different probability distributions defined on a measurement context $\{ a, b \}$ as Table~\ref{tab:Locality}. We can verify that $d_1$ and $d_2$ do agree on every local measurement event.
\begin{eqnarray*}
&& \left. d_1 \right|_a ((0)) = \left. d_1 \right|_a ((1)) = \left. d_1 \right|_b ((0)) = \left. d_1 \right|_b ((1)) \\
& = & \left. d_2 \right|_a ((0)) = \left. d_2 \right|_a ((1)) = \left. d_2 \right|_b ((0)) = \left. d_2 \right|_b ((1)) = 1/2
\end{eqnarray*}

\begin{table*}[tbp]
\caption{\label{tab:Empirical_PR}Empirical model on PR-box. This model is not realizable in quantum information systems, but it still satisfies the condition of a no-signaling empirical model.}
\begin{tabular}{c|c c c c}
& $(0,0)$ & $(0,1)$ & $(1,0)$ & $(1,1)$ \\
\colrule
$\{a_1, b_1\}$ & 1/2 &  0  &  0  & 1/2 \\
$\{a_1, b_2\}$ & 1/2 &  0  &  0  & 1/2 \\
$\{a_2, b_1\}$ & 1/2 &  0  &  0  & 1/2 \\
$\{a_2, b_2\}$ &  0  & 1/2 & 1/2 &  0 
\end{tabular}
\end{table*}
For the counter-example of gluing, the best example is actually the PR-box, although it is not a quantum realizable model. Table~\ref{tab:Empirical_PR} illustrates the empirical model of PR-box, which is not actually a quantum model but still satisfies the condition of being no-signaling empirical model. We can try gluing the first three rows to obtain the probability distribution over $X$, but in this case, there only exists a unique global distribution $d \in \mathcal{D}_R\mathcal{E}(X)$, such that $d((0,0,0,0)) = d((1,1,1,1)) = 1/2$ and 0 otherwise. However, it fails to match with the distribution over the context $\{ a_2, b_2 \}$ since there is no probability of getting $s(a_2) \neq s(b_2)$ from this global distribution.

In general, empirical models are compatible families of distributions. Although the sheaf condition does not hold for the entire presheaf $\mathcal{D}_R\mathcal{E}$, it is possible to ask if the gluing property holds for such a specific family $\left\{ e_C \right\} _{C \in \mathcal{M}}$. This is equivalent to asking the existence of a \textit{global distribution} $d \in \mathcal{D}_{R}\mathcal{E}(X)$, defined on the entire set of measurements $X$. Such a global distribution defines a distribution on the set $\mathcal{E}(X) = O^X$, which marginalizes to yield the behavior of the empirical model: \textit{i.e.} $\forall C \in \mathcal{M}, \left. d \right| _{C} = e_C$.

\begin{remark}
A global distribution is often called a \textit{global section} in a sheaf-theoretic sense. However, in this paper, we will keep the notation of event/assignment and distribution to distinguish one from another.
\end{remark}

Now, let's consider the meaning of the existence of a global distribution for a given empirical model $e$. Here we will show that $e$ has a global distribution if and only if it is realized by a \textit{factorisable hidden-variable model}.

\begin{definition}
Given a measurement cover $\mathcal{M}$, let $\Lambda$ be a set of values for a hidden variable. A \textit{hidden-variable model} $h$ over $\Lambda$ is an assignment such that:
\begin{itemize}
    \item $h_{\Lambda} \in \mathcal{D}_{R}(\Lambda)$;
    \item $\forall \lambda \in \Lambda$ and $C \in \mathcal{M}$, $h_{C}^{\lambda} \in \mathcal{D}_{R}\mathcal{E}(C)$;
    \item $\forall \lambda \in \Lambda$ and $C, C' \in \mathcal{M}$, $\left. h_{C}^{\lambda} \right| _{C \cap C'} = \left. h_{C'}^{\lambda} \right| _{C \cap C'}$.
\end{itemize}
\end{definition}

\begin{definition}
A hidden-variable model $h$ \textit{realizes} an empirical model $e$ if, $\forall C \in \mathcal{M}$ and $s \in \mathcal{E}(C)$:
\begin{equation*}
e_{C}(s) = \sum_{\lambda \in \Lambda} h_{C}^{\lambda}(s) \cdot h_{\Lambda}(\lambda).
\end{equation*}
\end{definition}

\begin{definition}
A hidden-variable model $h$ is \textit{factorisable} if, $\forall C \in \mathcal{M}$ and $s \in \mathcal{E}(C)$:
\begin{equation*}
h_{C}^{\lambda}(s) = \prod_{m \in C} \left. h_{C}^{\lambda} \right| _{m} (\left. s \right| _{m}).
\end{equation*}
\end{definition}

\begin{theorem}
\label{thm:Hidden_variable}
Let $e$ be an empirical model defined on a measurement cover $\mathcal{M}$ for a distribution functor $\mathcal{D}_{R}$. The following are equivalent.
\begin{enumerate}[(a)]
    \item $e$ has a global distribution.
    \item $e$ has a realization by a factorisable hidden-variable model.
\end{enumerate}
\end{theorem}
\begin{proof}
$(a) \rightarrow (b)$: For each $s \in O^{X}$, we define a distribution $\delta_{s} \in \mathcal{D}_{R}\mathcal{E}(X)$ such that $\delta_{s}(s') = 1$ for $s = s'$ and $0$ otherwise. Then if $e$ has a global distribution $d$:
\begin{equation*}
e_{C}(s) = \left. d \right| _{C} (s) = \sum_{s' \in \mathcal{E}(X), \left. s' \right| _{C} = s} d(s') = \sum_{s' \in \mathcal{E}(X)} \delta_{\left. s' \right| _{C}}(s) \cdot d(s').
\end{equation*}
Here, we identify $\mathcal{E}(X)$ to $\Lambda$, $d$ to $h_{\Lambda}$, and $\delta_{\left. s' \right| _{C}}$ to $h_{C}^{\lambda}$, so we showed $e$ is realized by a hidden variable model. Moreover, it is factorisable from that:
\begin{equation*}
\delta_{\left. s' \right| _{C}}(s) = \prod_{m \in C} \delta_{\left. s' \right| _{m}} (\left. s \right| _{m}).
\end{equation*}
$(b) \rightarrow (a)$: Suppose that $e$ is realized by a factorisable hidden-variable model $h$. For each $m \in X$, we define $h_{m}^{\lambda} := \left. h_{C}^{\lambda} \right| _{m} \in \mathcal{D}_{R}\mathcal{E}(m)$ for any $C \in \mathcal{M}$ such that $m \in C$. By the compatibility of the family $\left\{ h_{C}^{\lambda} \right\}$, this definition is independent of the choice C, being well defined. We define a distribution $h_{X}^{\lambda} \in \mathcal{D}_{R}\mathcal{E}(X)$ for each $\lambda \in \Lambda$ by:
\begin{equation*}
h_{X}^{\lambda}(s) = \prod_{m \in X} h_{m}^{\lambda}(\left. s \right| _{m}).
\end{equation*}
Now to show that $h_{X}^{\lambda}$ is a distribution, let the set of measurements $X$ be enumerated as $X = \{ m_1, \cdots, m_p \}$. Specify the global assignment $s \in \mathcal{E}(X)$ by a tuple $(o_1, \cdots, o_p)$, where $o_i = s(m_i)$. Then we can calculate:
\begin{eqnarray*}
&& \displaystyle \sum_{s \in \mathcal{E}(X)} \prod_{m \in X} h_{m}^{\lambda}(\left. s \right| _m) \\
= && \displaystyle \sum_{s \in \mathcal{E}(X)} \prod_{i = 1}^{p} h_{m_i}^{\lambda}(\left. s \right| _{m_i}) \\
= && \displaystyle \sum_{o_1} h_{m_1}^{\lambda}(m_1 \mapsto o_1) \cdot \left( \sum_{o_2} h_{m_2}^{\lambda}(m_2 \mapsto o_2) \cdot \left( \cdots \left( \sum_{o_p} h_{m_p}^{\lambda}(m_p \mapsto o_p) \right) \right) \right) \\
= && \displaystyle \sum_{o_1} h_{m_1}^{\lambda}(m_1 \mapsto o_1) \cdot \left( \sum_{o_2} h_{m_2}^{\lambda}(m_2 \mapsto o_2) \cdot \left( \cdots \left( 1 \right) \right) \right) \\
= && \cdots \\
= && \displaystyle \sum_{o_1} h_{m_1}^{\lambda}(m_1 \mapsto o_1) \cdot 1 = 1. 
\end{eqnarray*}
We can also show that for each context $C \in \mathcal{M}$, $\left. h_{X}^{\lambda} \right|_{C} = h_{C}^{\lambda}$. We choose an enumeration of $X$ such that $C = \{ m_1, \cdots, m_q\}$, $q \leq p$. Then:
\begin{eqnarray*}
\left. h_{X}^{\lambda} \right|_{C}(s)
= && \displaystyle \sum_{s' \in \mathcal{E}(X), \left. s' \right|_{C} = s} h_{X}^{\lambda}(s') \\
= && \displaystyle \sum_{s' \in \mathcal{E}(X), \left. s' \right|_{C} = s} \prod_{i=1}^{p} h_{m_i}^{\lambda}(m_i \mapsto o_i) \\
= && \displaystyle \prod_{i=1}^{q} h_{m_i}^{\lambda}(m_i \mapsto o_i) \cdot \left( 
\sum_{o_{q+1}, \cdots, o_p} \prod_{j=q+1}^{p} h_{m_j}^{\lambda}(m_j \mapsto o_j) \right) \\
= && h_{C}^{\lambda} \cdot 1 = h_{C}^{\lambda}. 
\end{eqnarray*}
Now we define a distribution $d \in \mathcal{D}_{R}\mathcal{E}(X)$ by averaging over the hidden variables:
\begin{equation*}
d(s):= \sum_{\lambda \in \Lambda} h_{X}^{\lambda}(s) \cdot h_{\Lambda}(\lambda).
\end{equation*}
The condition for distribution is automatically satisfied from that $\sum_{\lambda \in \Lambda} h_{\Lambda}(\lambda) = 1$. Moreover, $d$ restricts at each context $C$ to yield $e_C$ as:
\begin{eqnarray*}
\left. d \right|_{C}(s)
= && \displaystyle \sum_{s' \in \mathcal{E}(X), \left. s' \right|_{C} = s} d(s') \\
= && \displaystyle \sum_{s' \in \mathcal{E}(X), \left. s' \right|_{C} = s} \sum_{\lambda \in \Lambda} h_{X}^{\lambda}(s') \cdot h_{\Lambda}(\lambda) \\
= && \displaystyle \sum_{\lambda \in \Lambda} h_{\Lambda}(\lambda) \cdot \left. h_{X}^{\lambda} \right|_{C}(s) \\
= && \displaystyle \sum_{\lambda \in \Lambda} h_{\Lambda}(\lambda) \cdot h_{C}^{\lambda}(s) \\
= && e_{C}(s). 
\end{eqnarray*}
Thus $d$ is a global distribution for $e$.
\end{proof}

This result provides a definitive justification for equating the phenomena of non-locality and contextuality with obstructions to the existence of global distributions. We shall say the empirical model $e$ is \textit{noncontextual} if there exists a global distribution for $e$.

\subsection{Existence of global distributions and the contextual fraction}

\begin{definition}
Given a measurement cover $\mathcal{M}$, let $n:= \left| \mathcal{E}(X) \right|$ be the number of global assignments, and $m:= \sum_{C \in \mathcal{M}} \left| \mathcal{E}(C) \right| = \left| \left\{ \left< C, s \right> \mid C \in \mathcal{M}, s \in \mathcal{E}(C) \right\} \right|$ be the number of local assignments ranging over contexts. The \textit{incidence matrix} $\mathbf{M}$ is an $m \times n$ Boolean matrix that records the restriction relation between global and local assignments:
\begin{equation*}
\mathbf{M}\left[ \left< C, s \right> , g \right] :=
\begin{cases}
1 & \mathrm{if} \; \left. g \right|_C = s; \\
0 & \mathrm{otherwise.}
\end{cases}
\end{equation*}
\end{definition}

The incidence matrix conceptually represents the tuple of restriction maps
\begin{equation*}
\mathcal{E}(X) \rightarrow \prod_{C \in \mathcal{M}} \mathcal{E}(C)
:: s \mapsto \left( \left. s \right| _C \right) _{C \in \mathcal{M}} .
\end{equation*}
At the same time, viewing it as a matrix over the semiring $R$, it acts by matrix multiplication on distributions in $\mathcal{D}_R\mathcal{E}(X)$, represented as row vectors:
\begin{equation*}
d \mapsto \left( \left. d \right| _C \right) _{C \in \mathcal{M}} .
\end{equation*}
Thus, the image of this map will be the set of families $\left\{ e_C \right\} _{C \in \mathcal{M}}$ which arise from global distributions.

\begin{proposition}
Let $e$ be an empirical model defined on a measurement cover $\mathcal{M}$. Let $\mathbf{M}$ be an incidence matrix of $\mathcal{M}$. Let a vector $\mathbf{V}$ defined by $\mathbf{V}[i] = e_{C}(s_i)$, which encodes the value of the empirical model on each local assignment. Then, solutions to the following equation correspond bijectively to global distributions for $e$.
\begin{equation}
\label{eq:Linear_program}
\mathbf{M}\mathbf{X} = \mathbf{V} \quad \mathrm{where} \: \mathbf{X} \in R^n, \; \left\Vert \mathbf{X} \right\Vert_1 = 1.
\end{equation}
\end{proposition}

With this proposition, we can now compute the existence of global distributions from the given empirical model. However, the existence of global distributions depends on the semiring $R$, on which empirical models are discussed. We have naturally thought of probability distributions so far, which are defined on the non-negative reals $\mathbb{R}_{\geq 0}$, but here, let me address how the situation differs according to the semiring.

There are three main examples of semirings discussed in the field of contextuality: the Booleans
\begin{equation*}
(\mathbb{B}, \vee, \wedge),
\end{equation*}
the non-negative reals
\begin{equation*}
(\mathbb{R}_{\geq 0},+,\times),
\end{equation*}
and the reals
\begin{equation*}
(\mathbb{R},+,\times).
\end{equation*}
We call the distribution $d \in \mathcal{D}_\mathbb{B}$ on Boolean semiring a \textit{possibility distribution}. In the case of the semiring $\mathbb{R}_{\geq 0}$, $\mathcal{D}_{\mathbb{R}_{\geq 0}}$ is the set of \textit{probability distributions}. In the case of the reals $\mathbb{R}$, $\mathcal{D}_\mathbb{R}$ is the set of \textit{signed probability measures} with finite support, allowing for 'negative probabilities.' It can be shown that we can always find global distributions for the given empirical model on the semiring $\mathbb{R}$, making $\mathcal{D}_\mathbb{R}$ a sheaf. It is also straightforward to show that the probabilistic global distribution can exist only when the possibilistic global distribution exists.

Now we will extend this discussion about the incidence matrix further to define the measure of contextuality. Here we consider the semiring $\mathbb{R}_{\geq 0}$. First, define a convex sum $\lambda e + (1-\lambda) e'$ of two empirical models $e$ and $e'$ on the same measurement cover, by taking the convex sum of probability distributions on each context. Compatibility is preserved in the convex sum; hence, it yields a well-defined empirical model.
\begin{definition}
Given an empirical model $e$, consider a convex decomposition
\begin{equation}
\label{eq:Convex_decomp}
e = \lambda e^{NC} + (1-\lambda) e', \quad \lambda \in [0, 1],
\end{equation}
where $e^{NC}$ denotes a noncontextual model and $e'$ is another empirical model. The \textit{noncontextual fraction} $\mathsf{NCF}(e)$ is the maximum possible value of such $\lambda$. The \textit{contextual fraction} is given by $\mathsf{CF}(e) := 1 - \mathsf{NCF}(e)$.
\end{definition}

Then, we will consider the computation of this contextual fraction using the incidence matrix. Recall the equation \eqref{eq:Linear_program}, which is valid only for a noncontextual model. Since now we can no longer guarantee the solution of $\mathbf{M}\mathbf{X} = \mathbf{V}$, we will consider a relaxed condition:
\begin{equation*}
\mathbf{M}\mathbf{X} \leq \mathbf{V},
\end{equation*}
where the comparison is elementwise throughout this section. Let an empirical model $e$ and some convex decomposition $e = \lambda e^{NC} + (1-\lambda) e'$ be given. Because both empirical models are defined on the same measurement cover, the incidence matrix is the same for both models. Let $\mathbf{V}^{NC}$ encode the noncontextual model $e^{NC}$ and suppose a vector $\mathbf{X}^{NC} \in \mathbb{R}_{\geq 0}^n$ satisfies $\mathbf{M}\mathbf{X}^{NC} = \mathbf{V}^{NC}$, $\left\Vert \mathbf{X} \right\Vert_1 = 1$. Then, since $\mathbf{V} = \lambda \mathbf{V}^{NC} + (1-\lambda) \mathbf{V}'$,
\begin{equation*}
\lambda \mathbf{V}^{NC} = \mathbf{M} \lambda \mathbf{X}^{NC} \leq \mathbf{M} \left( \lambda \mathbf{X}^{NC} + (1-\lambda) \mathbf{X}' \right) = \mathbf{M}\mathbf{X} \leq \mathbf{V},
\end{equation*}
which implies existence of a solution $\mathbf{M}\mathbf{X} \leq \mathbf{V}$ with a lower bound $\lambda \mathbf{V}^{NC}$. Here, the norm $\left\Vert \mathbf{X} \right\Vert_1 \geq \left\Vert \lambda \mathbf{X}^{NC} + (1-\lambda) \mathbf{X}' \right\Vert_1 \geq \lambda$. Thus, an optimal solution $\mathbf{X}^*$ of the following linear programming(LP) gives the noncontextual fraction by $\mathsf{NCF}(e) = \max{\lambda} = \left\Vert \mathbf{X}^* \right\Vert_1$:
\begin{eqnarray*}
\begin{array}{ll}
\mathtt{Find} & \mathbf{X} \in \mathbb{R}^n \\
\mathtt{maximizing} & \left\Vert \mathbf{X} \right\Vert_1 \\
\mathtt{subject \; to} & \mathbf{M}\mathbf{X} \leq \mathbf{V}.
\end{array}
\end{eqnarray*}
Thus, we obtain a method computing a contextual fraction of a given empirical model $e$. One can refer to Ref \cite{ABM17} to check that this contextual fraction has monotonicity for the free operations of a resource theory, thus, works as a useful measure.

\subsection{Possibilistic empirical models and strong contextuality}

Once we have this measure of contextuality, one can notice the empirical models whose contextuality equals 1. In fact, an empirical model has a contextual fraction of 1 if and only if it is strongly contextual. Here we define strong contextuality based on a possibilistic empirical model and show the equivalence of maximally contextual and strongly contextual.

\begin{definition}[Possibilistic empirical model]
A \textit{possibilistic empirical model} $\mathcal{S}$ is a sub-presheaf of $\mathcal{E}$ such that:
\begin{itemize}
    \item every compatible family for the measurement cover $\mathcal{M}$ induces a global assignment;
    \item $\mathcal{S}$ is \textit{flasque beneath the cover}: If $U \subset U' \subset C \in \mathcal{M}$ then every $s \in \mathcal{S}(U)$ is the restriction of some $s' \in \mathcal{S}(U')$.
\end{itemize}
When an empirical model $e$ is given, a possibilistic empirical model $\mathcal{S}_e$ of $e$ can be derived as follows:
\begin{equation*}
\mathcal{S}_{e}(U) := \left\{ s \in \mathcal{E}(U) \mid \forall C \in \mathcal{M}, \left. s \right|_{U \cap C} \in \mathrm{supp}\left( \left. e_C \right|_{U \cap C} \right) \right\}.
\end{equation*}
\end{definition}
\begin{remark}
Here we define the possibilistic empirical model $\mathcal{S}_e$ of $e$ as a family of sets, rather than a family of possibilistic distributions $\mathcal{S}_{e, C}: C \rightarrow \mathbb{B}$. In fact, there is no problem with interpreting the possibilistic empirical model as a normal empirical model of a possibility distribution, but, we rather choose that the possibilistic empirical model returns a set of possible events, \textit{i.e.} $s \in \mathcal{S}_e(U)$ if $\left. s \right|_{U \cap C} \in \mathrm{supp}(\left. e_C \right|_{U \cap C})$ to clarify the notation.
\end{remark}

\begin{definition}
Let $\mathcal{S}$ be a possibilistic empirical model for a measurement scenario $\left< X, \mathcal{M}, O \right>$. We say that $\mathcal{S}$ is
\begin{itemize}
    \item \textit{logically contextual at $s \in \mathcal{S}(C)$} if there is no global assignment $g \in \mathcal{S}(X)$ such that $\left. g \right|_C = s$;
    \item \textit{logically contextual} if $\mathcal{S}$ is logically contextual at some local assignment. Otherwise, it is \textit{non-contextual};
    \item \textit{strongly contextual}, written $\mathrm{SC}(\mathcal{S})$, if $\mathcal{S}$ has no global assignment. \textit{i.e.} $\mathcal{S}(X) = \emptyset$.
\end{itemize}
\end{definition}
\begin{remark}
Note that possibilistic contextuality and strong contextuality differ from each other. An empirical model is possibilistically contextual if there is no global possibility distribution $d \in \mathcal{D}_R\mathcal{E}(X)$ and is strongly contextual if there is no global assignment $g \in \mathcal{E}(X)$.
\end{remark}

\begin{proposition}
An empirical model $e$ is strongly contextual if and only if it is maximally contextual.
\end{proposition}
\begin{proof}
Suppose that $e$ admits a convex decomposition \eqref{eq:Convex_decomp} as follows:
\begin{equation*}
e = \lambda e^{NC} + (1-\lambda) e', \quad \lambda \in [0, 1].
\end{equation*}
By the Theorem~\ref{thm:Hidden_variable}, we can take the non-contextual empirical model $e^{NC}$ to be a convex sum of deterministic models $\sum_i \mu_i \delta_{s_i}$, where each $s_i \in \mathcal{E}(X)$ is a global assignment. If $\lambda > 0$, then from \eqref{eq:Convex_decomp}, $s_i \in \mathcal{S}_e(X)$ for each $i$. Thus strong contextuality implies maximal contextuality.

For the converse, suppose that $s \in \mathcal{S}_e(X)$. Taking $e^{NC} = \lambda \cdot \delta_s$, we shall define $q$ such that \eqref{eq:Convex_decomp} holds. For each $C \in \mathcal{M}$ and $s' \in \mathcal{E}(C)$:
\begin{equation*}
q_C(s') := \frac{e_C(s') - \lambda \cdot \delta_{\left. s\right|_c}(s')}{1 - \lambda}.
\end{equation*}
It is easily verified that, for each $C$, $\sum_{s' \in \mathcal{E}(C)}$. To ensure that $q$ is always non-negative, we must have $\lambda \leq \inf_{C \in \mathcal{M}} e_C(\left. s \right|_C)$. Since this is the infimum of a finite set of positive numbers, we can find $\lambda > 0$ satisfying this condition.

It remains to verify that $q$ is no-signalling, \textit{i.e.} that $\left\{ q_C \right\}$ forms a compatible family. Given $C, C' \in \mathcal{M}$, fix $s_0 \in \mathcal{E} \left( C \cap C' \right)$. Now
\begin{equation*}
\left. q_C \right|_{C \cap C'} (s_0) = \frac{1}{1-\lambda} \left[ \left( \sum_{s' \in \mathcal{E}(C), \left. s' \right|_{C \cap C'} = s_0} e_C(s') \right) - \lambda \cdot \left. \delta_s \right|_{C \cap C'} (s_0) \right].
\end{equation*}
A similar analysis applies to $\left. q_{C'} \right|_{C \cap C'} (s_0)$. Using the compatibility of $e$, we conclude that $\left. q_C \right|_{C \cap C'} = \left. q_{C'} \right|_{C \cap C'}$
\end{proof}

\section{All-versus-nothing arguments and partial groups}

When we look into a quantum system, we see that the measurement depends on both the quantum state and the observable. However, it has been observed that some sets of observables inhere the contextuality independently from the quantum state. This type of contextuality, earlier observed by Kochen and Specker \cite{KS75}, has been developed to define different types of contextuality \cite{Mer93, ORBR17, KWB18, KL19}. Here, we formulate them with a sheaf-theoretic structure, starting from what is formally studied as an all-versus-nothing (AvN) argument \cite{ABK+15, ABCP17, Aas22}. We extend this argument to state-independent AvN and claim that Kochen-Specker type contextuality is, in fact, state-independent AvN in a partial closure.

\subsection{Mermin's square}

\begin{figure}
\includegraphics[width=0.3\textwidth]{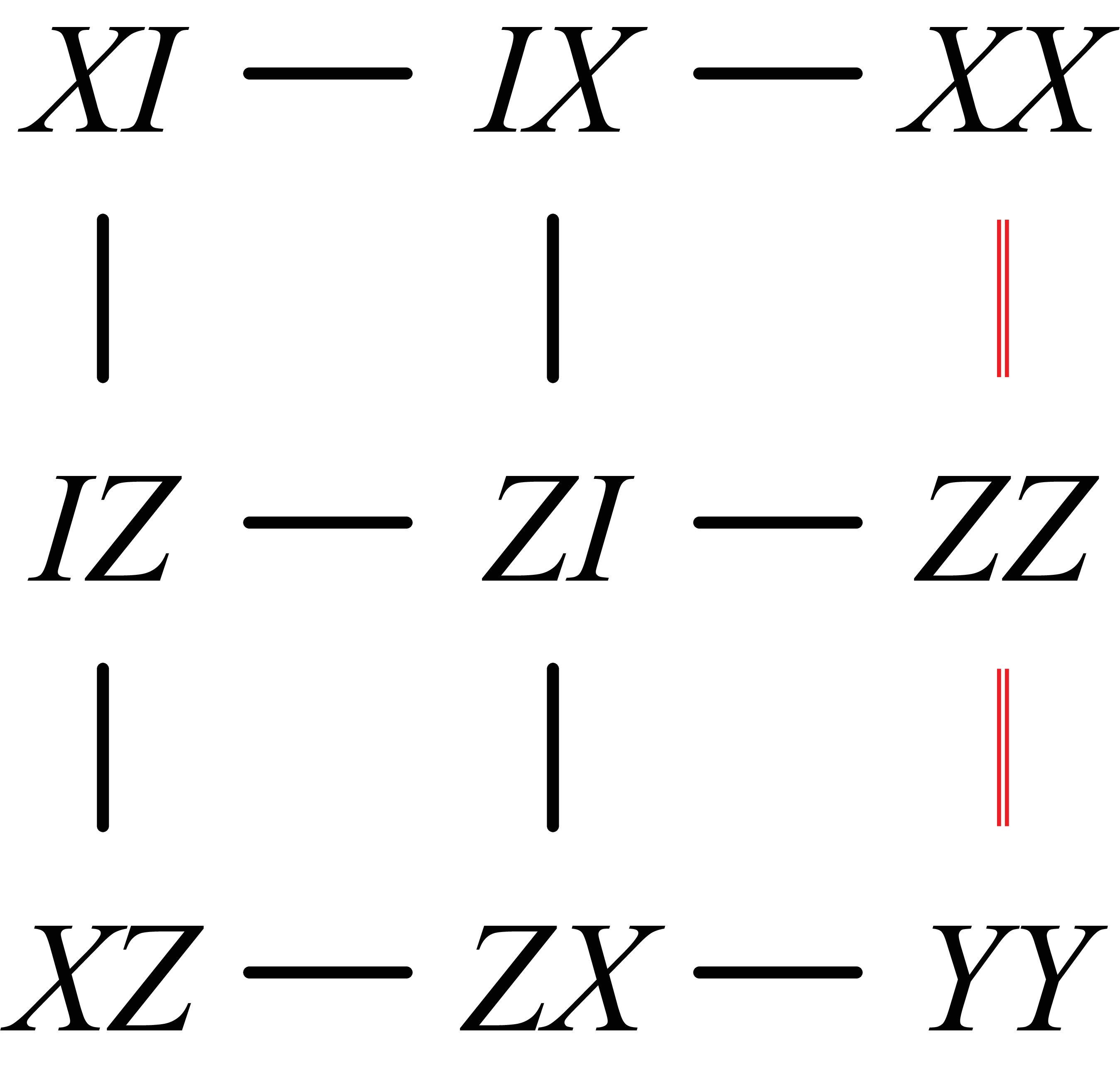}
\caption{\label{fig:Mermins_square}Mermin's square. The observables within each row and each column mutually commute, and the product of all three observables equals $+I$ except for the last column being $-I$.}
\end{figure}
Before we start, let's first look into the simple example of state-independent AvN. Fig.~\ref{fig:Mermins_square} shows the Mermin's square \cite{Mer93} that consists of mutually commuting Pauli observables on two qubits. In each row and each column, the product of the first two operators equals the last one, except for the last column: $XX \cdot ZZ = -YY$. Once we rearrange the product, we get that the product of every row and column equals $\pm I$: $+I$ for all three rows and the first two columns from the left, and $-I$ for the last column.

Here, consider a measurement scenario $\left< X, \mathcal{M}, \mathbb{Z}_2 \right>$ where each row and each column translates to a measurement context $C \in \mathcal{M}$. Then, try assigning a global assignment $g: X \rightarrow \mathbb{Z}_2$ to this measurement scenario. The product equations are mapped to the following linear equations,
\begin{equation}
\label{eq:Mermins_square}
\begin{split}
g(XI) \oplus g(IX) \oplus g(XX) & = 0, \\
g(IZ) \oplus g(ZI) \oplus g(ZZ) & = 0, \\
g(XZ) \oplus g(ZX) \oplus g(YY) & = 0, \\
g(XI) \oplus g(IZ) \oplus g(XZ) & = 0, \\
g(IX) \oplus g(ZI) \oplus g(ZX) & = 0, \\
g(XX) \oplus g(ZZ) \oplus g(YY) & = 1,
\end{split}
\end{equation}
where the outcome $g(x) \in \mathbb{Z}_2$ corresponds to the eigenvalue of $x$ by $\lambda(x) = (-1)^{g(x)}$. This system of linear equations is not satisfiable by any global assignment $g$. To see this, add the equations altogether then each element on the left-hand side is added twice, which ends up being 0, while the right-hand side equals 1. Therefore, no quantum state can realize a non-contextual empirical model with these observables, \textit{i.e.}, this set of observables inheres inconsistency.

\begin{table*}
\caption{\label{tab:Empirical_Mermins_square}Empirical model on Mermin's square realized by Bell state $\left| \phi^{+} \right>$.}
\begin{ruledtabular}
\begin{tabular}{c|cccccccc}
$C$ & \textbf{000} & \textbf{001} & \textbf{010} & \textbf{011} & \textbf{100} & \textbf{101} & \textbf{110} & \textbf{111} \\
\hline
$\{ XI, IX, XX \}$ & 1/2 &     &     &     &     &     & 1/2 &     \\
$\{ IZ, ZI, ZZ \}$ & 1/2 &     &     &     &     &     & 1/2 &     \\
$\{ XZ, ZX, YY \}$ &     &     &     & 1/2 &     & 1/2 &     &     \\
$\{ XI, IZ, XZ \}$ & 1/4 &     &     & 1/4 &     & 1/4 & 1/4 &     \\
$\{ IX, ZI, ZX \}$ & 1/4 &     &     & 1/4 &     & 1/4 & 1/4 &     \\
$\{ XX, ZZ, YY \}$ &     &  1  &     &     &     &     &     &     
\end{tabular}
\end{ruledtabular}
\end{table*}
For example, consider this measurement scenario realized by $\left| \phi^{+} \right> = \left( \left| 00 \right> + \left| 11 \right> \right) / 2$. Table~\ref{tab:Empirical_Mermins_square} illustrates the empirical model on Mermin's square realized by $\left| \phi^{+} \right>$. It is obvious that this empirical model is strongly contextual. In fact, we can prove that every empirical model with this measurement scenario is strongly contextual.

\subsection{Consistency of an \textit{R}-linear theory}

The key concept first to notice is a ring structure given to the set of outcomes $O$, which justifies writing a measurement scenario in the form $\left< X, \mathcal{M}, R \right>$, where $R$ is the ring of outcomes.

\begin{definition}
Let $\left< X, \mathcal{M}, R \right>$ be a measurement scenario. An \textit{$R$-linear equation} is a triple $\phi = \left< C, r, a \right>$ where $C \in \mathcal{M}$ is a context, $r: C \rightarrow R$ assigns a coefficient in $R$ to each $x \in C$, and $a \in R$ is a constant. A local assignment $s \in \mathcal{E}(C)$ \textit{satisfies} $\phi$, written $s \vDash \phi$, if
\begin{equation*}
\sum_{x \in C} r(x) s(x) = a.
\end{equation*}
An \textit{$R$-linear theory} $\mathbb{T}_R$ is a set of $R$-linear equations. A global assignment $g \in \mathcal{E}(X)$ satisfies $\mathbb{T}_R$, written $g \vDash \mathbb{T}_R$, if $\forall \phi = \left< C, r, a \right> \in \mathbb{T}_R, \left. g \right|_C \vDash \phi$. $\mathbb{T}_R$ is \textit{consistent} if there exists a global assignment $g$ that satisfies $\mathbb{T}_R$.
\end{definition}
\begin{remark}
Note that the consistency of an $R$-linear theory is defined on an event sheaf $\mathcal{E}$, not on an empirical model.
\end{remark}

For example, consider Mermin's square. The linear equation $s(XI) \oplus s(IX) \oplus s(XX) = 0$ for a local assignment $s \in \mathcal{E}(C)$ on the context $C = \{ XI, IX, XX \} \in \mathcal{M}$ translates to a triple $\left< \{ XI, IX, XX \}, 1, 0 \right>$ where 1 denotes a constant function that maps every measurement $x$ to 1. The explicit statement of the whole linear theory $\mathbb{T}_{\mathbb{Z}_2}$ of Mermin's square is given as follows:
\begin{eqnarray*}
\mathbb{T}_{\mathbb{Z}_2} = \left\{ \left< \{ XI, IX, XX \}, 1, 0 \right>, \left< \{ IZ, ZI, ZZ \}, 1, 0 \right>, \left< \{ XZ, ZX, YY \}, 1, 0 \right>, \right. \\
\left. \left< \{ XI, IZ, XZ \}, 1, 0 \right>, \left< \{ IX, ZI, ZX \}, 1, 0 \right>, \left< \{ XX, ZZ, YY \}, 1, 1 \right> \right\} ,
\end{eqnarray*}
Once we suppose a global assignment $g \in \mathcal{E}(X)$ such that $\forall \phi \in \mathbb{T}_{\mathbb{Z}_2}, \left. g \right|_C \vDash \phi$. This deduces the same equation with \eqref{eq:Mermins_square}, thus, a contradiction. Therefore, the linear theory of Mermin's square is inconsistent.

\subsection{AvN arguments}

An AvN argument \cite{ABK+15, ABCP17, Aas22} is characterized by an $R$-linear theory of a possibilistic empirical model. Here, we consider an $R$-linear theory $\mathbb{T}_R$ derived from a set of assignments $S \subset \mathcal{E}(U)$, as follows:
\begin{equation*}
\mathbb{T}_R(S) := \left\{ \phi \mid \forall s \in S, s \vDash \phi \right\} .
\end{equation*}
Once we have a possibilistic empirical model $\mathcal{S}$, $\mathcal{S}(C)$ gives a set of possible assignments. Thus, we can connect a possibilistic empirical model $\mathcal{S}$ to an $R$-linear theory.

\begin{definition}
Given a possibilistic empirical model $\mathcal{S}$, an \textit{$R$-linear theory $\mathbb{T}_R(\mathcal{S})$} is defined by:
\begin{equation*}
\mathbb{T}_R(\mathcal{S}) := \bigcup_{C \in \mathcal{M}} \mathbb{T}_R \left( \mathcal{S}(C) \right) = \left\{ \phi = \left< C, r, a \right> \mid \forall s \in \mathcal{S}(C), s \vDash \phi \right\}.
\end{equation*}
We say that $\mathcal{S}$ is $\mathrm{AvN}$ if its $R$-linear theory $\mathbb{T}_R(\mathcal{S})$ is inconsistent. i.e. there is no global assignment $g \in \mathcal{E}(X)$ such that $g \vDash \mathbb{T}_R(\mathcal{S})$.
\end{definition}

\begin{table*}
\caption{\label{tab:Probabilistic_Mermins_square}Probabilistic empirical model on Mermin's square realized by Bell state $\left| \phi^{+} \right>$.}
\begin{ruledtabular}
\begin{tabular}{c|cccccccc}
$C$ & \textbf{000} & \textbf{001} & \textbf{010} & \textbf{011} & \textbf{100} & \textbf{101} & \textbf{110} & \textbf{111} \\
\hline
$\{ XI, IX, XX \}$ & 1 &   &   &   &   &   & 1 &   \\
$\{ IZ, ZI, ZZ \}$ & 1 &   &   &   &   &   & 1 &   \\
$\{ XZ, ZX, YY \}$ &   &   &   & 1 &   & 1 &   &   \\
$\{ XI, IZ, XZ \}$ & 1 &   &   & 1 &   & 1 & 1 &   \\
$\{ IX, ZI, ZX \}$ & 1 &   &   & 1 &   & 1 & 1 &   \\
$\{ XX, ZZ, YY \}$ &   & 1 &   &   &   &   &   &   
\end{tabular}
\end{ruledtabular}
\end{table*}
For example, go back to the empirical model on Mermin's square realized by Bell state $\left| \phi^+ \right>$. The possibilistic empirical model obtained from the probabilistic empirical model is characterized in Table~\ref{tab:Probabilistic_Mermins_square}. The theory $\mathbb{T}_R(\mathcal{S})$ contains the following linear equations.
\begin{eqnarray*}
& s(XI) \oplus s(IX) = 0, ~ s(XX) = 0 & \mathrm{on} ~ \{ XI, IX, XX \} \\
& s(IZ) \oplus s(ZI) = 0, ~ s(ZZ) = 0 & \mathrm{on} ~ \{ IZ, ZI, ZZ \} \\
& s(XZ) \oplus s(ZX) = 1, ~ s(YY) = 1 & \mathrm{on} ~ \{ XZ, ZX, YY \} \\
& s(XI) \oplus s(IZ) \oplus s(XZ) = 0 & \mathrm{on} ~ \{ XI, IZ, XZ \} \\
& s(IX) \oplus s(ZI) \oplus s(ZX) = 0 & \mathrm{on} ~ \{ IX, ZI, ZX \} \\
& s(XX) = 0, ~ s(ZZ) = 0, ~ s(YY) = 1 & \mathrm{on} ~ \{ XX, ZZ, YY \}
\end{eqnarray*}
Each local assignment $s$ is defined on the specific context in each line. Note that this theory contains more linear equations compared to the previous example. Again, once we suppose a global assignment $g$, we find an inconsistency, so $\mathcal{S}$ is $\mathrm{AvN}$.

\begin{proposition}
If $\mathcal{S}$ is $\mathrm{AvN}$ then $\mathcal{S}$ is strongly contextual.
\end{proposition}
\begin{proof}
Suppose $\mathcal{S}$ is not strongly contextual, i.e. that there is some $g \in \mathcal{S}(X)$. Then for each $\phi = \left< C, r, a \right> \in \mathbb{T}_R(\mathcal{S})$, $\left. g \right|_{C} \in \mathcal{S}(C)$, hence $\left. g \right|_{C} \vDash \phi$ by the definition of $\mathbb{T}_R(\mathcal{S})$. Thus, $\mathbb{T}_R(\mathcal{S})$ is consistent.
\end{proof}

\subsection{State-independent AvN}

While we defined the AvN argument on a possibilistic empirical model, Mermin's square seems to lie in a stronger class of contextuality, as we have seen that the linear theory of Mermin's square with a possibilistic empirical model is larger than that without an empirical model. When we have a measurement $x$ in a quantum system, it is actually associated with an observable $O$ and a state $\psi$, or more generally, a density matrix $\rho$, so to characterize a measurement $x = O_\psi$. Hereby, we focus on an algebra given by the set of observables, and each measurement $x$ is specified by an observable $O$. In the remaining part of section 3, we denote the set of observables as $X = \{ x_i \}$, which is realized to be a set of measurements $X_\psi = \{ x_{i, \psi} \}$ by a state $\psi$, and a measurement cover of observables $\mathcal{M}$, realized to be a measurement cover $\mathcal{M}_\psi$ by a state $\psi$. We may abuse the word ``measurement cover'' for the measurement cover of observables, but clearly, the measurement cover of observables must be realized by a state $\psi$ to be an actual measurement cover.

Now turning to the set of observables $X$, we first restrict our concern to the subset of Pauli $n$-group $G_n$. This is to take care before dealing with the general quantum observables, \textit{i.e.}, arbitrary self-adjoint operators on a complex Hilbert space, where Tsirelson's problem \cite{Tsi87} may arise. To be specific, when we consider a span of a set of nonlocal observables, it may not be able to approximate the generated operator algebra with a set of local Pauli operators, even if the dimension of the Hilbert space is finite \cite{Vid21, JNV+22}. Thus, we restrict our concern to Pauli groups $G_n$ where each basis operator is a tensor product of local Pauli operators.

Now in $G_n$, we have a multiplicative group operation, which is generally not commutative. However, an event sheaf $\mathcal{E}$ still can be assigned on this set of Pauli observables. For any subset $U \subset X \subset G_n$, we consider an assignment $s: U \rightarrow \mathbb{Z}_2$ that maps to the measurement of each observable $x \in U$ as $\lambda(x) = \left( -1 \right)^{s(x)}$ for its eigenvalue $\lambda(x) = \pm 1$.

Once we have a multiplicative relation in $X$, we observe that it translates to a linear equation according to the following rules:
\begin{itemize}
    \item when $\prod_i x_i = I$ holds for mutually commuting $x_i \in X$, the relation corresponds to a linear equation $\sum_{i} s(x_i) = 0$;
    \item when $\prod_i x_i = -I$ holds for mutually commuting $x_i \in X$, the relation corresponds to a linear equation $\sum_{i} s(x_i) = 1$.
\end{itemize}

For example, when we have $x \cdot y = z$, it is equivalent to $x \cdot y \cdot z^{-1} = I$, so we can map it to a linear equation. In the $G_n$, $x^{-1} = x$, so it further reduces to an equation $x \cdot y \cdot z = I$ in $G_n$. The second relation stands for the dual element $-x$ of $x$ such that $x \cdot (-x) = -I$. This is dual in the sense of a local assignment $s$ where $s(x) \oplus s(-x) = 1$ always holds.

Note that the condition of mutual commutation is required to map multiplication in $X$ to addition in $\mathbb{Z}_2$, which is commutative. However, there is an intriguing object we already have that encodes the commutability, namely, the measurement cover $\mathcal{M}$.

From the definition of a measurement cover, $\mathcal{M}$ is a family of maximal commuting subsets of $X$, \textit{i.e.}, for any $x, y \in C \subset X$, $x \cdot y = y \cdot x$, and if there exists $C \subset C' \subset X$ for $C \in \mathcal{M}$ such that the elements of $C'$ commutes with each other, then $C = C'$. Since commutative multiplication only occurs in a commuting subset, it only occurs in a measurement context. This means that we can define a linear theory of a set of observables $X$, as presented in the following definition.

\begin{definition}[State-independent AvN]
Given a set of observables $X \subset G_n$, \textit{a linear theory $\mathbb{T}_{\mathbb{Z}_2}(X)$} is defined by:
\begin{equation*}
\mathbb{T}_{\mathbb{Z}_2}(X) := \left\{ \phi = \left< C, r, a \right> \mid C \in \mathcal{M}, \prod_{x \in C} x^{r(x)} = (-1)^{a} \cdot I \right\}.
\end{equation*}
We say that $X$ is \textit{state-independently} AvN if its linear theory $\mathbb{T}_{\mathbb{Z}_2}(X)$ is inconsistent. i.e. there is no $g \in \mathcal{E}(X)$ such that $g \vDash \mathbb{T}_{\mathbb{Z}_2}(X)$.
\end{definition}

\begin{proposition}
If $X$ is state-independently $\mathrm{AvN}$ then any possibilistic model $\mathcal{S}$ on $\left< X_\psi, \mathcal{M}_\psi, \mathbb{Z}_2 \right>$ is $\mathrm{AvN}$.
\end{proposition}
\begin{proof}
What we want to show is that any possibilistic empirical model $\mathcal{S}$ realized by any state satisfies $\mathbb{T}_{\mathbb{Z}_2}(X)$. Here we deal with a density matrix $\rho$ for generality. For each linear equation $\phi = \left< C, r, a \right> \in \mathbb{T}_{\mathbb{Z}_2}(X)$, what we want to have is:
\begin{equation*}
\sum_{x \in C} r(x) s(x_\rho) = a,
\end{equation*}
for $\forall s \in \mathcal{S}(C)$ realized by an arbitrary density matrix $\rho$. Here we show it by contradiction.

Suppose not. Then it means $\exists s \in \mathcal{S}(C)$ such that $\sum_{x \in C} r(x) s(x_\rho) \neq a$. This equation with the assignment $s$ maps to an equation with an assignment of eigenvalues $\lambda(x) = (-1)^{s(x_\rho)}$ as follows:
\begin{equation*}
\prod_{x \in C} \lambda(x)^{r(x)} \neq (-1)^a.
\end{equation*}
When $s \in \mathcal{S}(C)$, it means that $\mathrm{Prob}(\lambda) \neq 0$, where $\mathrm{Prob}(\lambda)$ is obtaining eigenvalues $\lambda(x)$ sequentially going through each $x \in C$. Now look into the updated state:
\begin{equation*}
\rho' = \left( \prod_{x \in C} P_{\lambda(x)} \right) \rho \left( \prod_{x \in C} P_{\lambda(x)} \right),
\end{equation*}
where $P_{\lambda(x)}$ is the projector of each eigenvalue $\lambda(x)$ of $x$. We have all $x$'s in $C$, so they mutually commute, and so do their projectors $P_{\lambda(x)}$, so the expression is valid. Then, from the definition of $\mathbb{T}_{\mathbb{Z}_2}(\mathcal{M})$, we have $\prod_{x \in C'} x = (-1)^a \cdot I$. Then we can try measuring this from the updated state $\rho$ as follows:
\begin{equation*}
\mathrm{Tr} \left( \left( \prod_{x \in C} x^{r(x)} \right) \rho' \right) = (-1)^a \cdot \mathrm{Tr} \left( \rho' \right) = (-1)^a.
\end{equation*}
Here we derive a contradiction:
\begin{eqnarray*}
\mathrm{lhs} & = & \mathrm{Tr} \left( \left( \prod_{x \in C} x^{r(x)} \right) \left( \prod_{x \in C} P_{\lambda(x)} \right) \rho \left( \prod_{x \in C} P_{\lambda(x)} \right) \right) \\
& = & \mathrm{Tr} \left( \left( \prod_{x \in C} x^{r(x)} P_{\lambda(x)} \right) \rho \left( \prod_{x \in C} P_{\lambda(x)} \right) \right) \\
& = & \mathrm{Tr} \left( \left( \prod_{x \in C}\lambda(x)^{r(x)} \cdot P_{\lambda(x)} \right) \rho \left( \prod_{x \in C} P_{\lambda(x)} \right) \right) \\
& = & \mathrm{Tr} \left( \left( \prod_{x \in C} \lambda(x)^{r(x)} \right) \cdot \left( \prod_{x \in C} P_{\lambda(x)} \right) \rho \left( \prod_{x \in C} P_{\lambda(x)} \right) \right) \\
& = & \left( \prod_{x \in C} \lambda(x)^{r(x)} \right) \mathrm{Tr} \left( \rho' \right) \\
& = & \prod_{x \in C} \lambda(x)^{r(x)} \neq (-1)^a = \mathrm{rhs}.
\end{eqnarray*}

This is because $\mathrm{Tr}(\rho') \neq 1$, which means $\prod_{x \in C} P_{\lambda(x)} = 0$, thus $\mathrm{Prob}(\lambda) = 0$, resulting in that $s \notin \mathcal{S}$. Hence we have $\forall s \in \mathcal{S}(C)$, $s \vDash \phi$, therefore $\phi \in \mathbb{T}_{\mathbb{Z}_2}(\mathcal{S})$. Since we have shown it with arbitrary $\phi \in \mathbb{T}_{\mathbb{Z}_2}(X)$, $\mathbb{T}_{\mathbb{Z}_2}(X) \subset \mathbb{T}_{\mathbb{Z}_2}(\mathcal{S})$. However, since $\mathbb{T}_{\mathbb{Z}_2}(X)$ is not satisfiable, $\mathbb{T}_{\mathbb{Z}_2}(\mathcal{S})$ is also not satisfiable.
\end{proof}

Now, we highlight that this definition catches the idea discussed by Kochen and Specker \cite{KS75}. Here, $X$ can be translated to a partial algebra in Kochen and Specker's argument where the commeasurability corresponds to the measurement cover $\mathcal{M}$. The rule of mapping the multiplicative equations of $X$ into a linear theory $\mathbb{T}_{\mathbb{Z}_2}(X)$ coincides with an embedding of $X$ into a Boolean algebra. The inconsistency of the linear theory implies the failure of embeddability, which was previously connected to the non-classical logic of quantum mechanics.

However, there still exists some ambiguity in matching these concepts to Kochen and Specker's argument. In particular, the set $X \subset G_n$ is not necessarily closed under multiplication, which disturbs the direct interpretation of $X$ as a partial algebra as in Kochen and Specker's work. Considering that some useful arguments on quantum advantage, \textit{e.g.}, Refs.~\onlinecite{KWB18, KL19}, are based on Kochen-Specker type contextuality, we should characterize the relation between the set of observables and a partial algebra, by defining the measurement cover $\mathcal{M}$ of $X$ in a more rigorous way.

\subsection{Partial group and Kochen-Specker type contextuality}

First, let's start with the concept of a partial group. While the idea of partial algebra was raised earlier by Kochen and Specker in 1975 \cite{KS75}, the definition of a partial group is discussed only after Assiry in 2018 \cite{Ass18}. Here we define an abelian partial group based on those two studies without exactly showing how this definition connects to the aforementioned mathematical objects.

\begin{definition}
An \textit{abelian partial group $\left< \Gamma, \smallfrown, \cdot \right>$} consists of a set $\Gamma$, a binary relation $\smallfrown \subset \Gamma \times \Gamma$, and a binary operation $\cdot: \smallfrown \rightarrow \Gamma$, satisfying the following properties:
\begin{itemize}
    \item the relation $\smallfrown$ is reflexive and symmetric, \textit{i.e.}, for any $a, b \in \Gamma$, $a \smallfrown a$ and $a \smallfrown b \Leftrightarrow b \smallfrown a$;
    \item the relation $\smallfrown$ is closed under the operation $\cdot$, \textit{i.e.}, if $a \smallfrown b$, $b \smallfrown c$ and $c \smallfrown a$ then $a \cdot b \smallfrown c$;
    \item there exists an identity $1 \in \Gamma$ such that $1 \smallfrown a$ for any $a \in \Gamma$ and $1 \cdot a = a \cdot 1 = a$;
    \item if a subset $S$ of $\Gamma$ satisfies that $a \smallfrown b$ for any $a, b \in S$, $S$ generates an abelian group $\left< S \right> \subset \Gamma$.
\end{itemize}
Such relation $\smallfrown$ is called \textit{commutativity} or \textit{commeasurability} of $\Gamma$. The binary operation $\cdot: \smallfrown \rightarrow \Gamma$ is in fact a \textit{partial binary operation on $\Gamma$}, written by $\cdot: \Gamma \times \Gamma \rightharpoonup \Gamma$. Note that for a commuting subset $S \subset \Gamma$, the partial binary operation $\cdot$ on $\Gamma$ becomes a commutative binary operation on the abelian group $\left< S \right> \subset \Gamma$. Meanwhile, we may refer to an abelian partial group $\Gamma$ without specifying its commutativity and partial binary operation.
\end{definition}

For example, the set of observables in Mermin's square (including $\pm$ sign and an identity $II$) forms an abelian partial group where the commutativity corresponds to the usual commutativity relation of observables, and the partial binary operation is defined on a commuting set of observables. Remark the following properties in this example:
\begin{itemize}
    \item $\Gamma = \pm \{ II, XI, IX, XX, IZ, ZI, ZZ, XZ, ZX, YY\}$;
    \item $XI \smallfrown IZ$, but $XI \not\smallfrown ZI$;
    \item $S = \{ XI, IZ \}$ is a commuting subset of $\Gamma$ and generates an abelian group $\left< S \right> = \{ II, XI, IZ, XZ \} \subset \Gamma$.
\end{itemize}

\begin{remark}
A group $G$ can be interpreted as an abelian partial group where the commutativity corresponds to the usual commutativity relation of the elements of $G$, and the partial binary operation is derived from the binary operation of $G$. Thus, we shall say a group $G$ is an abelian partial group without specifying its commutativity relation and partial binary operation explicitly.
\end{remark}

This definition of an abelian partial group characterizes the algebraic structure of quantum observables where the multiplication of observables is restricted by commutativity so that we can bring the hidden variable arguments to each commuting subset. Furthermore, it can extend to a partial algebra by introducing addition and scalar multiplication, so that it can coincide with an operator algebra, as in Kochen and Specker's original work.

Kochen and Specker found that the hidden variable model and the quantum model branch out when we try to glue the measurement outcomes of overlapping commuting subsets to obtain the global hidden variable model, \textit{i.e.}, there exists a quantum model of measurement inconsistent with the hidden variable model on the global domain, but still satisfies the classical arguments on each commuting domain. This concept justifies the statement: the quantum model of measurement depends on a ``context.''

While this partial algebraic structure plays an important role in quantum theory, we do notice that a set of observables $X$ in a measurement scenario is not generally a closed abelian partial group. For example, an $XZ-(2,2,2)$ scenario has an observable set $X = \{ XI, IX, ZI, IZ \}$ which is not closed under multiplication, \textit{e.g.}, $XI \cdot IX = XX \notin X$. This is exactly the difference between sheaf-theoretic contextuality and Kochen-Specker type contextuality: Kochen-Specker type contextuality is defined on an abelian partial group while sheaf-theoretic contextuality does not require the measurement set to be an abelian partial group.

Having said that, we can try linking those two definitions by considering the generation of an abelian partial group by the given set of measurements, or in other words, a partial closure of the given set.
\begin{definition}[Partial closure]
For a given subset $X$ of $\Gamma$, a \textit{partial closure $\overline{X}$ of $X$ on $\Gamma$} is the smallest abelian partial group in $\Gamma$ containing $X$.
\end{definition}

For example, $X_1 = \{ XI, IZ \}$ is a subset of a Pauli group $G_2$ of which a partial closure $\overline{X_1} = \{ II, XI, IZ, XZ\}$ is an abelian partial group in $G_2$. Likewise, the partial closure of $X_2 = \{ XI, IX, ZI, IZ \}$ is $\overline{X_2} = \pm \{ II, XI, IX, XX, IZ, ZI, ZZ, XZ, ZX, YY\} \subset G_2$.

Finally, we can characterize the measurement cover $\mathcal{M}$ of a set $X$ in the context of an abelian partial group of quantum observables. $\mathcal{M}$ is defined on an abelian partial group $G_n \supset X$, where each context $C \in \mathcal{M}$ is a maximal commuting subset of $X$. Similarly, we can derive the measurement cover $\overline{\mathcal{M}}$ of $\overline{X}$, where each context $C \in \overline{\mathcal{M}}$ is a maximal abelian group in $\overline{X}$. One may interpret $\mathcal{M}$ as a family of the intersections of $X$ and the abelian groups $C \in \overline{\mathcal{M}}$, which also catches the idea that the measurement cover $\mathcal{M}$ originates in the commutativity relation of $G_n \supset X$.

Now, it is straightforward to specify Kochen-Specker type contextuality in the language of an abelian partial group.
\begin{definition}[Kochen-Specker type contextuality]
Given a set $X \subset G_n$, $X$ has \textit{Kochen-Specker type contextuality} if there is no global assignment of eigenvalues $\lambda: \overline{X} \rightarrow \{ \pm 1 \}$ consistent with the partial commutative multiplication in $\overline{X}$.
\end{definition}

Kochen-Specker type contextuality is also called as \textit{contextuality in a closed sub-theory}, concerning that the argument only deals with a part of the quantum observables. It is also clear that Kochen-Specker type contextuality is nothing but the state-independent AvN in a partial closure, once we define the state-independent AvN in a partial closure as follows:
\begin{definition}
Given a set of observables $X \subset G_n$, $X$ is \textit{state-independently $\mathrm{AvN}$ in a partial closure} if $\mathbb{T}_{\mathbb{Z}_2}(\overline{X})$ is inconsistent.
\end{definition}

Following the argument so far, we state the following corollary without giving proof.
\begin{corollary}
A set of observables $X$ has Kochen-Specker type contextuality if and only if $X$ is state-independently $\mathrm{AvN}$ in a partial closure.
\end{corollary}

\begin{figure*}
\includegraphics[width=0.7\textwidth]{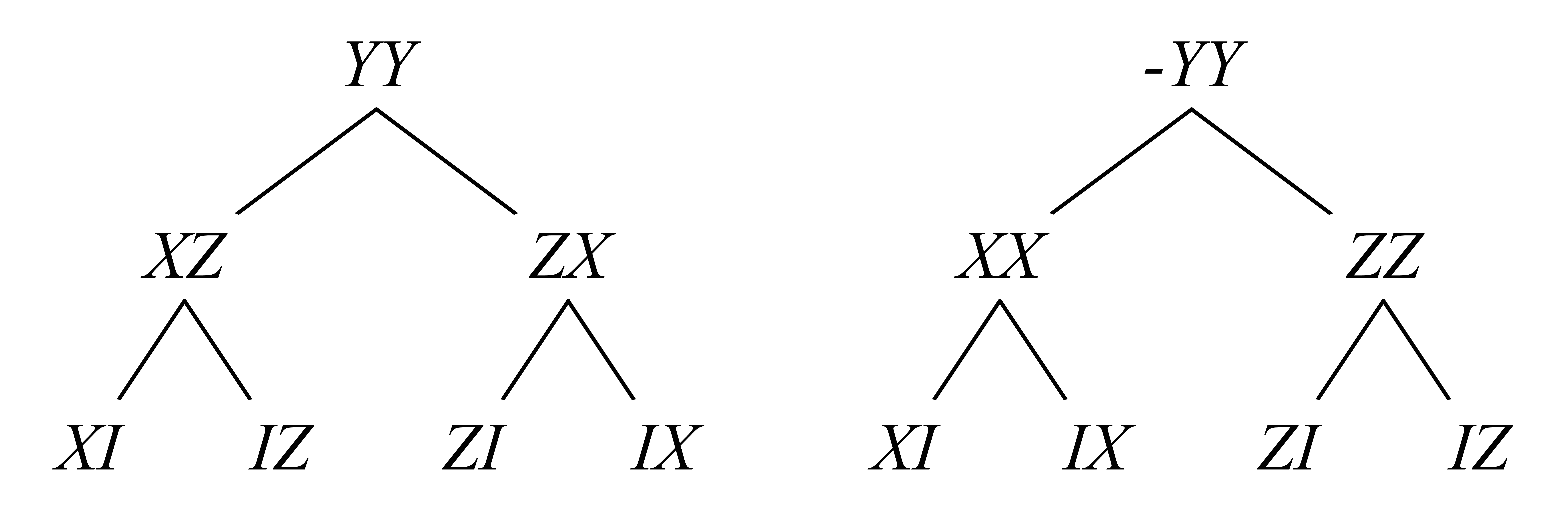}
\caption{\label{fig:Determining_tree}Determining tree proposed by Kirby and Love \cite{KL19}. Each parent operator and its children operators pairwise commute and the product of children operators equals the parent operator.}
\end{figure*}
Kirby and Love \cite{KL19} developed this concept to define their own way of witnessing contextuality, and applied it to evaluate the classical simulatability of practical quantum algorithms. Fig.~\ref{fig:Determining_tree} illustrates the example of the formalism Kirby and Love developed.

\begin{definition}
A \textit{determining tree $\tau_x$} for a Pauli observable $x$ over a set of Pauli observables $X$ is a tree whose nodes are Pauli operators and whose leaves are operators in $X$, such that:
\newpage
\begin{itemize}
    \item the root is $x$;
    \item all children of any particular parent pairwise commute as operators;
    \item every parent node is the operator product of its children.
\end{itemize}
We say that $x$ is determined by $X$ if there exists a determining tree $\tau_x$ over $X$. For a determining tree $\tau_x$, the \textit{determining set $D(\tau_x)$} is defined to be the set containing one copy of each operator with odd multiplicity as a leaf in $\tau_x$.
\end{definition}

Note that $x \in \overline{X}$ if and only if there exists $\tau_x$ over $X$. The determining tree images inductive production of $\overline{X}$. Whenever we have $\prod_i x_i = x$, $\left( \prod_i x_i \right) \cdot x = I$, which maps to a linear equation in $\mathbb{T}_{\mathbb{Z}_2}(\overline{X})$. Thus, each determining tree $\tau_x$ gives a linear theory $\phi_{\tau_x} \in \mathbb{T}_{\mathbb{Z}_2}(\overline{X})$.

Now, once we get a determining tree $\tau_x$, it pushes us to try inducing an assignment $g_{\tau_x}$ for $x \in \overline{X}$ from some global assignment $g \in \mathcal{E}(X)$, which assigns value $g_{\tau_x} = \sum_{x' \in D(\tau_x)}g(x')$ following the determining tree $\tau_x$ over $X$. However, what we can observe is that the determining tree for $x$ is not unique, and in fact, $g_{\tau_x}(x) \neq g_{\tau_x'}(x)$ in general. The following theorem characterizes this failure of assigning values to the determining tree.

\begin{theorem}[KL contextuality]
\label{thm:KL_contexutality}
A set $X$ of Pauli operators is state-independently $\mathrm{AvN}$ in a partial closure if and only if there exists a determining tree $\tau_x$ over $X$ and a determining tree $\tau_{-x}'$ over $X$ such that $D(\tau_x) = D(\tau_{-x}')$.
\end{theorem}

\begin{figure*}
\includegraphics[width=0.6\textwidth]{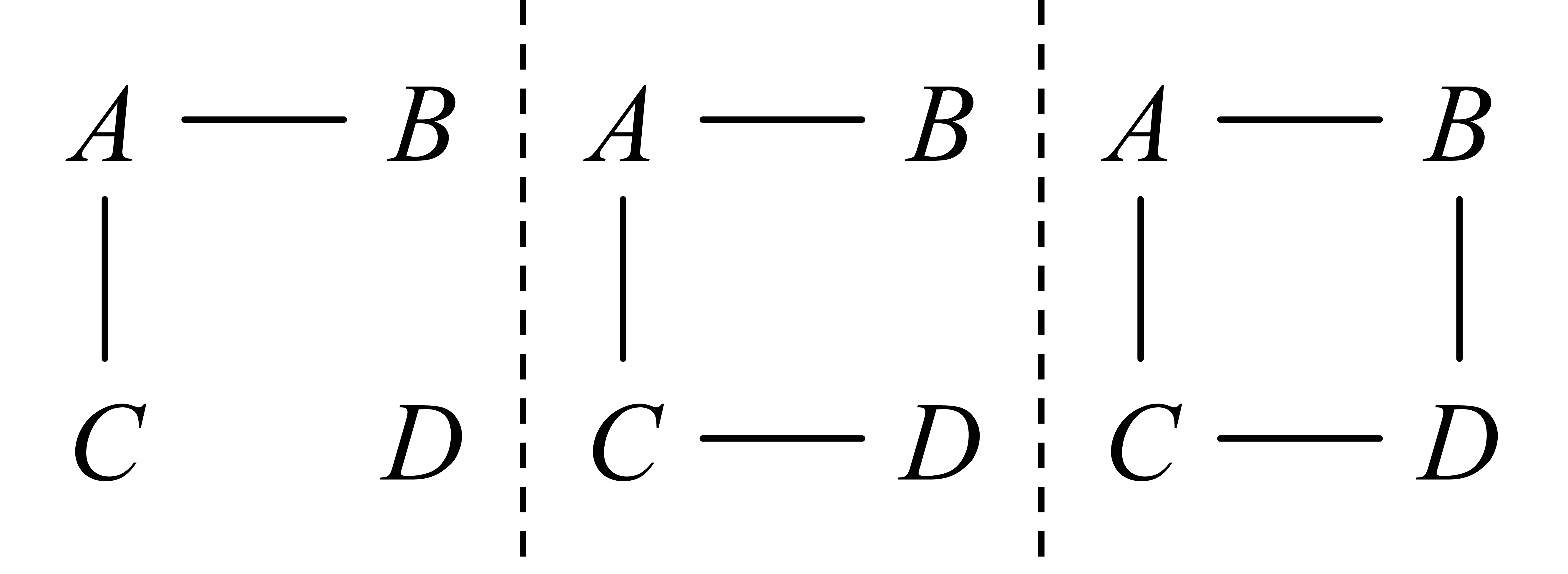}
\caption{\label{fig:Commutability_graph}Commutability graphs that induce state-independent AvN in a partial closure.}
\end{figure*}
\begin{corollary}
\label{cor:Commutability_graph}
A set $X \subset G_n$ is state-independently $\mathrm{AvN}$ in a partial closure if and only if it contains a subset consisting of four operators whose commutability graph has one of the forms given in Fig.~\ref{fig:Commutability_graph} (up to permutations of the operators).
\end{corollary}

We refer to Ref.~\onlinecite{KL19} for the proof of corollary~\ref{cor:Commutability_graph}. The conversed direction of theorem~\ref{thm:KL_contexutality} is easily confirmed. Suppose that there exists such a pair of determining trees, and suppose $g \in \mathcal{E}(X)$ induces a valid $g_\mathcal{T} \in \mathcal{E}(\overline{X})$ such that $\mathcal{T}_x = \tau_x$ and $\mathcal{T}_{-x} = \tau_{-x}'$. Then $g_\mathcal{T}(x) = \sum_{x' \in D(\tau_x)}g(x') = \sum_{x' \in D(\tau_{-x}')}g(x') = g_\mathcal{T}(-x)$, since $D(\tau_x) = D(\tau_{-x}')$. However, since $\pm x \in \overline{X}$, $\{ \pm I, \pm x \} \in \overline{\mathcal{M}}$, we have an equation $x (-x) = -I$, which leads to that $g_\mathcal{T}(x) + g_\mathcal{T}(-x) = 1$, which is a contradiction. Thus, no such assignment $g_\mathcal{T}$ is valid.

On the other hand, the forward direction of the theorem is not really obvious, although we do conjecture it to be true as it was in Ref.~\onlinecite{KL19}. Here we leave this proof for our future work, but we do state it as a theorem rather than a conjecture.

\subsection{Classification of contextuality}

\begin{figure*}
\includegraphics[width=0.7\textwidth]{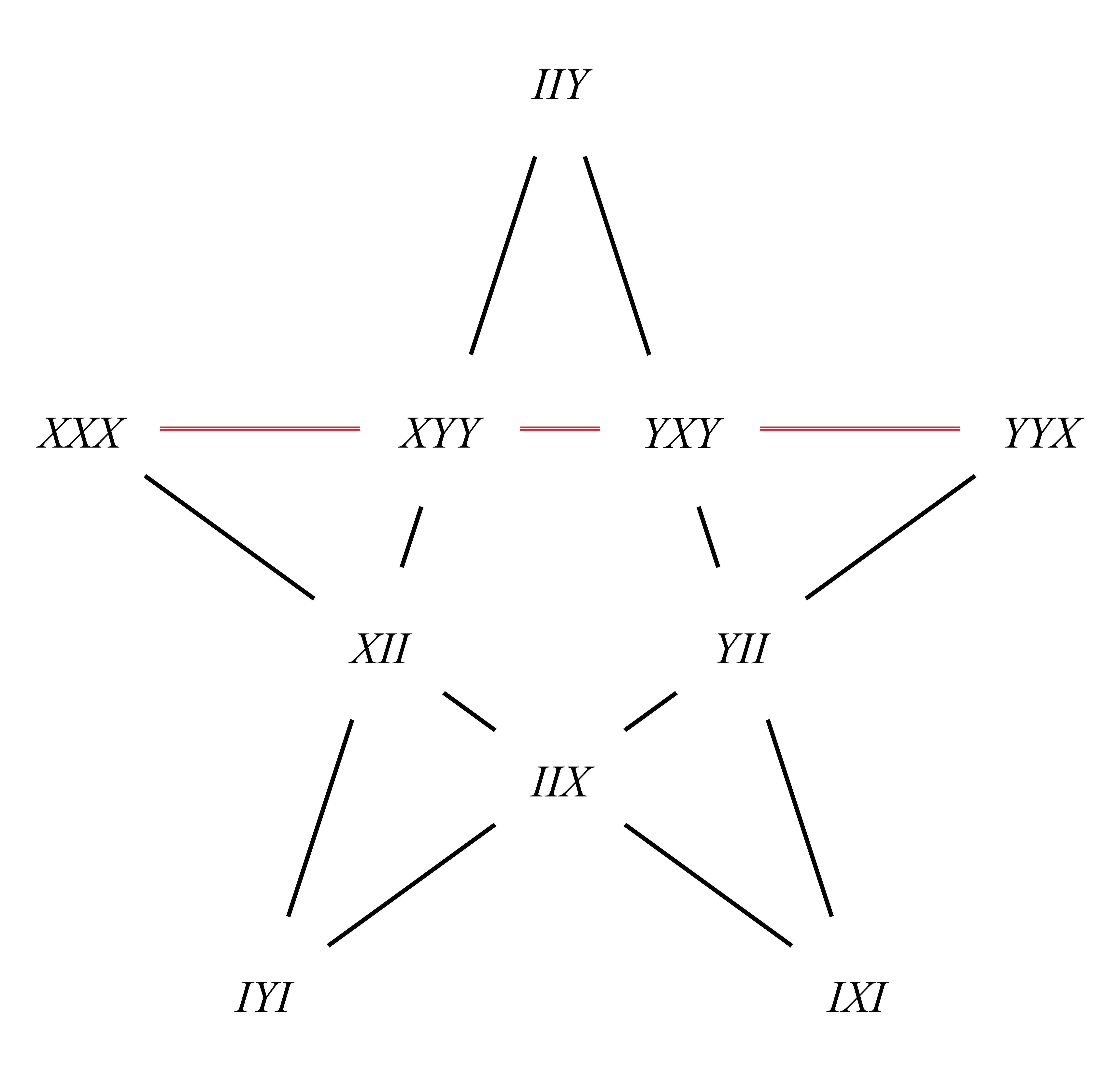}
\caption{\label{fig:Mermins_star}Mermin's star. The observables lying on each line mutually commute, and the product of all four observables equals $+I$ except for the horizontal line being $-I$.}
\end{figure*}
Although a given empirical model is simply contextual, or even non-contextual, it might have AvN contextuality in a partial closure. For example, consider the $\mathrm{XZ}-(2,2,2)$ Bell scenario. This is actually the case that motivated this research, where the Bell scenario with measurements $(XI, ZI, IX, IZ)$ is not contextual at all, but from Mermin's square, it turns out to have state-independent AvN in a partial closure. Hereby, the argument of Ref.~\onlinecite{ABC+18}, that the 2-qubit scenario is unable to show strong non-locality, seems to disagree with this result. However, we could realize that the statement is still correct since it considers non-locality, where only local observables are concerned.

\begin{table*}
\caption{\label{tab:Empirical_GHZ}Empirical model on $XY-(3,2,2)$ scenario realized by GHZ state.}
\begin{ruledtabular}
\begin{tabular}{c|cccccccc}
$C$ & \textbf{000} & \textbf{001} & \textbf{010} & \textbf{011} & \textbf{100} & \textbf{101} & \textbf{110} & \textbf{111} \\
\hline
$\{ XII, IXI, IIX \}$ & 1/4 &     &     & 1/4 &     & 1/4 & 1/4 &     \\
$\{ XII, IYI, IIY \}$ &     & 1/4 & 1/4 &     & 1/4 &     &     & 1/4 \\
$\{ YII, IXI, YII \}$ &     & 1/4 & 1/4 &     & 1/4 &     &     & 1/4 \\
$\{ YII, IYI, IIX \}$ &     & 1/4 & 1/4 &     & 1/4 &     &     & 1/4 \\
       $\vdots$       &     &     &     &\vdots&     &     &     &     
\end{tabular}
\end{ruledtabular}
\end{table*}
\begin{table*}
\caption{\label{tab:Empirical_equal}Empirical model on $XY-(3,2,2)$ scenario realized by equal superposition state.}
\begin{ruledtabular}
\begin{tabular}{c|cccccccc}
$C$ & \textbf{000} & \textbf{001} & \textbf{010} & \textbf{011} & \textbf{100} & \textbf{101} & \textbf{110} & \textbf{111} \\
\hline
$\{ XII, IXI, IIX \}$ & 1/8 & 1/8 & 1/8 & 1/8 & 1/8 & 1/8 & 1/8 & 1/8 \\
$\{ XII, IYI, IIY \}$ & 1/8 & 1/8 & 1/8 & 1/8 & 1/8 & 1/8 & 1/8 & 1/8 \\
$\{ YII, IXI, YII \}$ & 1/8 & 1/8 & 1/8 & 1/8 & 1/8 & 1/8 & 1/8 & 1/8 \\
$\{ YII, IYI, IIX \}$ & 1/8 & 1/8 & 1/8 & 1/8 & 1/8 & 1/8 & 1/8 & 1/8 \\
       $\vdots$       &     &     &     &\vdots&     &     &     &     
\end{tabular}
\end{ruledtabular}
\end{table*}
Another point we make is that the state-independent AvN class is strictly smaller than the ordinary AvN class. In other words, there exist state-dependent AvN models. For example, consider Mermin's star illustrated in Fig.~\ref{fig:Mermins_star}. Table~\ref{tab:Empirical_GHZ} shows the empirical model on $XY-(3,2,2)$ scenario realized by GHZ state, which is AvN. However, in the same measurement setting, we can realize another empirical model as in Table~\ref{tab:Empirical_equal} with equal superposition state $\left| +++ \right>$. Here, every probability in the scenario is $1/8$, which turns out to be non-contextual.

\begin{figure*}
\includegraphics[width=0.8\textwidth]{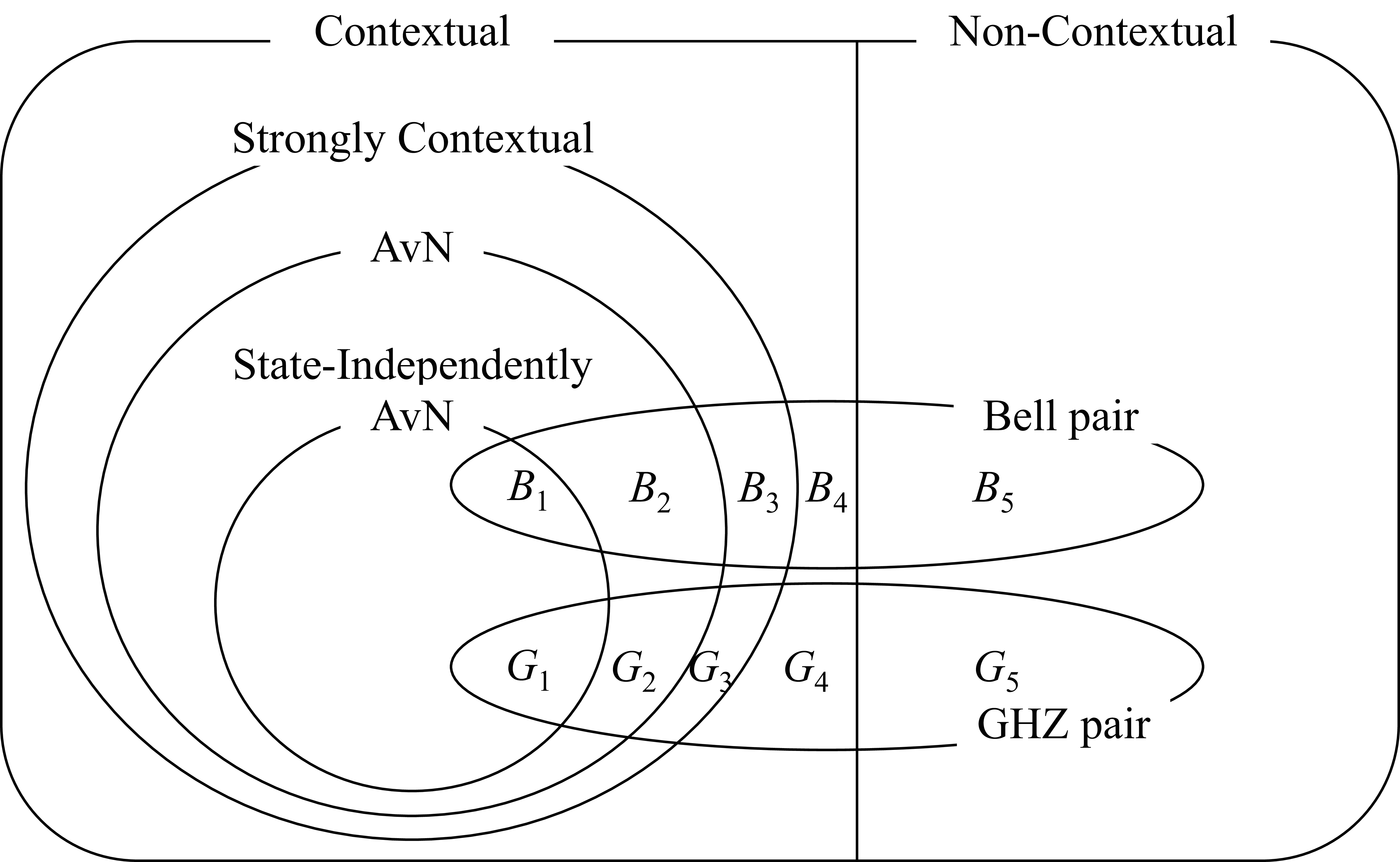}
\caption{\label{fig:State-side_view}State-side view. With the Bell pair, Mermin's square lies in $B_1$, $\mathrm{CHSH}-(2,2,2)$ Bell scenario lies in $B_4$, and $XZ-(2,2,2)$ scenario lies in $B_5$. With the GHZ state, Mermin's star lies in $G_1$, $XY-(3,2,2)$ scenario lies in $G_2$, and $XZ-(3,2,2)$ scenario lies in $G_5$. The areas $B_2, B_3, G_3, G_4$ seem to be abandoned, yet explicit proof is not addressed.}
\end{figure*}
\begin{figure*}
\includegraphics[width=0.8\textwidth]{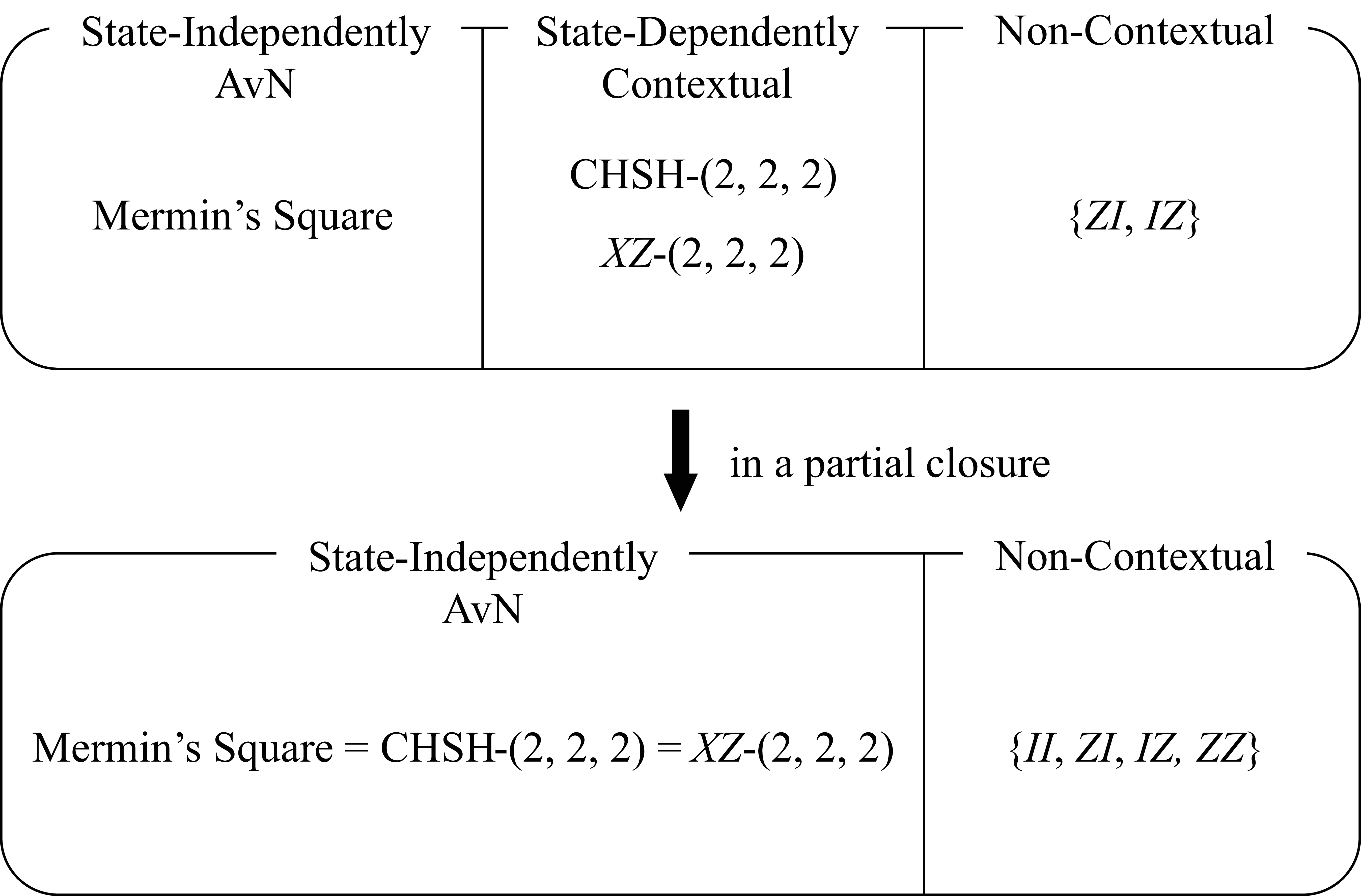}
\caption{\label{fig:Operator-side_view}Operator-side view with conjecture~\ref{conj:AvN_partial_closure}. Regarding 2-qubit scenarios, Mermin's square is state-independently AvN, and $\mathrm{CHSH}-(2,2,2)$ and $XZ-(2,2,2)$ scenarios are state-dependently contextual. The non-contextual example may given by $\mathcal{M} = \{ ZI, IZ \}$. However, in a partial closure, Mermin's square, $\mathrm{CHSH}-(2,2,2)$, and $XZ-(2,2,2)$ are all identical scenarios.}
\end{figure*}
As a result, we can summarize the relationship between different definitions of contextuality as Fig.~\ref{fig:State-side_view}. However, the picture seems quite different when it comes to the operator-side view, as we illustrate in Fig.~\ref{fig:Operator-side_view}. Although we know few cases of state-dependently contextual scenarios, they seem to integrate into the state-independent AvN when their partial closures are considered. In this view, we conjecture that state-dependent contextuality is in fact state-independent AvN in a partial closure. We present this idea in the following conjecture.
\begin{conjecture}
\label{conj:AvN_partial_closure}
Any measurement cover $\mathcal{M}$ that realizes contextual empirical model for some state $\psi$ is state-independently $\mathrm{AvN}$ in a partial closure.
\end{conjecture}

This conjecture yields an interesting idea that the partial closure may provide a way to connect sheaf-theoretical contextuality to state-independent contextuality. Here we leave the proof of this conjecture for our future work.

\newpage

\section{Conclusion}

In this report, we reviewed the sheaf-theoretic framework of contextuality and proposed state-independent AvN arguments. This work provides a coherent mathematical structure to compare each class of contextuality, clarifying the hierarchy of state-independent AvN - AvN - strong contextuality - contextuality. Kochen-Specker type contextuality integrates into this framework by considering a partial closure of the given set of measurements.

This work also develops the idea that contextuality does not necessarily require measurements to be ``local.'' While the compatibility condition of events and distributions originates in the no-signaling principle, the condition is still valid when we include non-local observables in the set of measurements. However, it requires a cautious approach to deal with non-local observables because of Tsireleson's problem \cite{Tsi87, Vid21, JNV+22}, which is also discussed in Ref.~\onlinecite{CLS17}. Here we restricted our concern to a Pauli $n$-group to avoid this problem.

Whilst we organized a consistent framework of contextuality, it cannot be affirmed that this framework is the most effective framework for formalizing contextuality in every case. A graph-theoretic approach based on Kochen-Specker type contextuality may be suitable for some proofs of quantum advantage, or a topological framework may be more effective in other cases. However, this work, together with Aasnæss's thesis \cite{Aas22}, implies a certain relationship between those approaches, thus enabling the translation of arguments in each framework from one to another.

Future studies may work on proving that state-dependent contextuality merges into state-independent AvN when its partial closure is concerned. It would also be a question if there is a set of observables that is state-independent contextual but not state-independent AvN, or an empirical model that is strongly contextual but not AvN. The generalization of state-independent contextuality to arbitrary self-adjoint operators on a complex Hilbert space would be another intriguing problem. Such discussions will clarify the point where classical and quantum information systems diverge.

In summary, this work presents an extensive approach from a sheaf-theoretic framework to Kochen-Specker type contextuality and state-independent contextuality, providing a consistent mathematical language to compare notions of contextuality. This will serve as a key tool to evaluate quantum advantage, guiding how the contextuality arguments from different frameworks can be translated to each other.

\bibliographystyle{apsrev4-2}
\bibliography{bib}

\end{document}